\documentclass[11pt]{article}

\usepackage[T1]{fontenc}
\usepackage[utf8]{inputenc}
\usepackage{lmodern}
\usepackage{amsmath,amssymb,amsthm,mathtools,bm}
\usepackage{geometry}
\usepackage{microtype}
\usepackage{needspace}
\usepackage{booktabs}
\usepackage{graphicx}
\usepackage{enumitem}
\usepackage{float}
\usepackage[numbers,sort&compress]{natbib}
\usepackage{authblk}
\usepackage{hyperref}

\hypersetup{colorlinks=true,linkcolor=blue,citecolor=blue,urlcolor=blue}
\newtheorem{theorem}{Theorem}[section]
\newtheorem{corollary}[theorem]{Corollary}
\newtheorem{proposition}[theorem]{Proposition}
\newtheorem{lemma}[theorem]{Lemma}
\theoremstyle{remark}

\newcommand{\Tr}{\operatorname{Tr}}
\newcommand{\Herm}{\operatorname{Herm}}
\newcommand{\Sym}{\operatorname{Sym}}
\newcommand{\End}{\operatorname{End}}
\newcommand{\SEP}{\mathrm{SEP}}
\newcommand{\op}{\mathrm{op}}
\newcommand{\id}{\mathrm{id}}
\newcommand{\cR}{\mathcal{R}}
\newcommand{\cH}{\mathcal{H}}
\newcommand{\cK}{\mathcal{K}}
\newcommand{\cM}{\mathcal{M}}
\newcommand{\cU}{\mathcal{U}}
\newcommand{\cV}{\mathcal{V}}
\newcommand{\cS}{\mathcal{S}}
\newcommand{\Ymn}{\mathbb{Y}_{m,n}}

\title{\textbf{Schur--Weyl Reduction of Finite-Copy Precision Bounds\\
for Multiparameter Quantum Estimation}}
\author[1]{Jianchao Zhang\thanks{Email: \texttt{jianchaozhang@cuhk.edu.cn}}}
\author[1,2,3]{Masahito Hayashi\thanks{Email: \texttt{hmasahito@cuhk.edu.cn}}}
\affil[1]{School of Data Science, The Chinese University of Hong Kong, Shenzhen, Longgang District, Shenzhen 518172, China}
\affil[2]{International Quantum Academy, Futian District, Shenzhen 518048, China}
\affil[3]{Graduate School of Mathematics, Nagoya University, Nagoya, 464-8602, Japan}

\date{October 28, 2026}

\begin{document}
\maketitle

\begin{abstract}
In multiparameter quantum estimation, incompatible measurements can impose trade-offs among estimation errors. Understanding the precision attainable with a finite number of jointly measured copies is therefore important, but direct optimization becomes costly as the tensor product space grows exponentially. For independent and identically distributed (i.i.d.) copies, the quantum statistical model is invariant under permutations of the copies. We exploit this symmetry through Schur--Weyl duality to remove permutation multiplicity spaces while using their dimensions as weights. This gives a reduced statistical model equivalent to the original under arbitrary collective measurements. On this reduced model, we derive an exact conic formulation of the tight Cram\'er--Rao bound and a semidefinite program for its Nagaoka--Hayashi (NH) relaxation. For fixed single-copy dimension, the reduced operator representation grows polynomially rather than exponentially with the number of copies. To approximate the tight bound, we use an exact quadratic identity to correct finite-ray dual solutions and obtain rigorous upper and lower bounds. For fixed copy number and parameter dimension, sufficiently fine meshes admit feasible primal points, and dual solutions with bounded objective gaps have uniformly bounded compact-search radii. These results yield quadratic convergence in the projective covering radius. Numerical examples illustrate the quadratic scaling of the mesh correction term and the computational reach of the reduced NH formulation.
\end{abstract}

\section{Introduction}
\label{sec:introduction}

Quantum parameter estimation has a long statistical foundation~\cite{Helstrom1976,Holevo2011,BraunsteinCaves1994}, and quantum sensing and quantum metrology apply this framework by optimizing the available probe and measurement resources~\cite{Degen2017,Szczykulska2016}. A central challenge in multiparameter quantum estimation is the incompatibility of measurements optimized for different parameters and the resulting trade-offs among estimation errors~\cite{Ragy2016,Albarelli2020,Demkowicz2020,Kull2020}. The allowed size of collective measurements constrains the attainable trade-offs among estimation errors~\cite{Chen2022PRA}. The number $m$ of jointly measured copies is therefore a natural finite resource. We consider the independent and identically distributed (i.i.d.) state family $\rho_\theta^{\otimes m}$, consisting of $m$ copies of a quantum state $\rho_\theta$, and allow arbitrary collective measurements on each block. Understanding this finite-copy hierarchy is necessary for connecting experimentally restricted measurements with asymptotic quantum-estimation theory~\cite{GillMassar2000,Hayashi2005}.

The Holevo--Nagaoka bound provides an asymptotic precision benchmark for multiparameter quantum estimation, and its attainability under collective measurements has been established for broad classes of quantum statistical models~\cite{Holevo2011,Yamagata2013,Yang2019}. For a fixed number of jointly measured copies, however, the achievable precision is not generally determined by this asymptotic result. Finite-copy state estimation has shown benefits of collective measurements~\cite{MassarPopescu1995,ZhuHayashi2018}, while studies of multiparameter estimation examine how trade-offs depend on the number of copies measured jointly~\cite{Chen2022PRL,Chen2022PRA} and whether gaps from asymptotic limits persist~\cite{Conlon2022Gap}. Experimental work has also demonstrated the advantages and feasibility of collective measurements in quantum sensing and state tomography~\cite{Roccia2018,Hou2018,Conlon2023}. On the computational side, semidefinite optimization has been used to evaluate the Holevo precision benchmark~\cite{AlbarelliFrielDatta2019}, while conic methods characterize the exact finite-copy precision and semidefinite relaxations yield computable lower bounds~\cite{HayashiOuyang2023,Conlon2021}. Nevertheless, direct computations become increasingly demanding as the number of copies grows, owing to the rapidly increasing size of the quantum system. Obtaining high numerical accuracy for the exact precision bound presents an additional challenge, even with existing approximation methods that provide rigorous error guarantees~\cite{HayashiOuyang2023}.

Although the Hilbert space of an $m$-copy quantum system grows exponentially with $m$, an i.i.d. quantum statistical model possesses substantial permutation symmetry. Exploiting this symmetry allows the tensor product representation to be reduced to smaller blocks, substantially simplifying many-copy calculations. Schur--Weyl duality provides a systematic framework for such reductions and has been applied in quantum state estimation and related statistical problems~\cite{KeylWerner2001,Bagan2006,BlumeKohout2013,MarvianSpekkens2014}. Related symmetry-based methods have also been used to simplify optimization in quantum metrology~\cite{Liu2023,Liu2024}. In this work, we exploit this structure to obtain reduced formulations of finite-copy multiparameter precision bounds.

Our first result establishes an exact Schur--Weyl reduction of the tight Cram\'er--Rao bound and its Nagaoka--Hayashi (NH) relaxation for finite-copy multiparameter estimation. We remove permutation multiplicity spaces and use their dimensions as weights to construct a reduced statistical model that is equivalent to the original $m$-copy model under arbitrary collective measurements. We combine this statistical equivalence with the conic framework of Hayashi and Ouyang~\cite{HayashiOuyang2023}. This gives a reduced conic program for the tight bound and a reduced semidefinite program (SDP) for the NH bound. Both reduced programs have the same optimal values as their respective full-space programs. The reduction also decreases the real dimension of the Hermitian operator space needed to represent permutation-invariant operators. For a single-copy Hilbert space of dimension $n$, the full Hermitian operator space on $m$ copies and its permutation-invariant subspace have dimensions, respectively,
\begin{equation}
 N_{\mathrm{full}}(m,n)=n^{2m},
 \qquad
 N_{\mathrm{op}}(m,n)=\binom{m+n^2-1}{n^2-1}.
\label{eq:intro-size}
\end{equation}
The second identity follows from Schur--Weyl duality~\cite{Sagan2001} and is derived in terms of its representation blocks in Sec.~\ref{subsec:size-reduction}. Thus, for fixed $n$, the dimension of the permutation-invariant Hermitian operator space grows polynomially rather than exponentially with $m$, while the finite-copy precision bounds remain unchanged.

Our second result improves the numerical approximation of the tight Cram\'er--Rao bound. Even after Schur--Weyl reduction, its conic formulation involves constraints that must hold for infinitely many directions. A natural approach is to select finitely many directions, leading to a finite semidefinite program~\cite{HayashiOuyang2023}. We exploit an exact quadratic identity to construct a correction to this finite-mesh approximation, obtaining rigorous upper and lower bounds on the tight value. For each fixed copy number and parameter dimension, we show that every sufficiently fine projective mesh admits a feasible primal point. We also establish a uniform bound on the compact-search radii of near-optimal sampled dual solutions. Together with the quadratic correction, these results yield an $O(\delta^2)$ gap between the finite-ray upper bound and a repaired dual lower bound, where $\delta$ is the mesh covering radius. This improves the earlier $O(\delta)$ estimate of Hayashi and Ouyang~\cite{HayashiOuyang2023}.

Numerical examples examine the quadratic correction for the tight bound and illustrate the large-copy computational reach of the reduced Nagaoka--Hayashi formulation. We also apply these methods to noisy qubit sensing models, including two-axis rotations and joint phase--damping estimation. Section~\ref{sec:sw-reduction} introduces the finite-copy statistical model and its Schur--Weyl reduction; Sec.~\ref{sec:reduced-bounds} derives the reduced optimization formulations; Sec.~\ref{sec:quadratic} develops the quadratic approximation; Sec.~\ref{sec:numerics} presents numerical results; Sec.~\ref{sec:applications} gives noisy sensing applications; and Sec.~\ref{sec:discussion} discusses the findings and limitations.
\section{Finite-copy multiparameter estimation and Schur--Weyl reduction}
\label{sec:sw-reduction}

\subsection{Finite-copy local statistical model}
\label{subsec:local-model}

Let $\mathcal S(\cH)$ denote the density operators on an $n$-dimensional Hilbert space $\cH$. We consider a finite-dimensional $d$-parameter quantum statistical model
\[
 \cM=\{\rho_\theta:\theta\in\Theta\subset\mathbb R^d\},
 \qquad
 \rho_\theta\in\mathcal S(\cH),
 \qquad
 \dim\cH=n.
\]
We work locally at an operating point translated to $\theta=0$ and use the one-copy state and tangent operators
\[
 \rho:=\rho_0,
 \qquad
 D_i:=\left.\frac{\partial\rho_\theta}{\partial\theta_i}\right|_{\theta=0},
 \qquad i=1,\ldots,d.
\]
We assume throughout that the model is differentiable in a neighborhood of the operating point, that $\rho$ is full rank ($\rho\succ0$), and that $D_1,\ldots,D_d$ are linearly independent. Let $G\in\mathbb R^{d\times d}$ be symmetric and positive definite.

For a fixed copy number $m$, the i.i.d. model is $\rho_\theta^{\otimes m}$ on $\cH^{\otimes m}$, with tangent operators
\begin{equation}
 D_i^{(m)}
 :=\left.\frac{\partial\rho_\theta^{\otimes m}}{\partial\theta_i}\right|_{\theta=0}
 =\sum_{r=1}^m\rho^{\otimes(r-1)}\otimes D_i\otimes\rho^{\otimes(m-r)}.
\label{eq:m-copy-tangents}
\end{equation}
The $m$-copy tangents $D_1^{(m)},\ldots,D_d^{(m)}$ remain linearly independent; a short proof is given in Lemma~\ref{app:tangent-independence} in Appendix~\ref{app:sw-details}.

A POVM $\Pi=\{\Pi_x\}$ on $\cH^{\otimes m}$ and estimator $\widehat\theta(x)\in\mathbb R^d$ are locally unbiased at $\theta=0$ if
\begin{align}
 \sum_x\widehat\theta_i(x)\Tr[\rho^{\otimes m}\Pi_x]&=0,\label{eq:local-unbiased-mean}\\
 \sum_x\widehat\theta_i(x)\Tr[D_j^{(m)}\Pi_x]&=\delta_{ij}.\label{eq:local-unbiased-derivative}
\end{align}
The first condition fixes the local mean and the second fixes the first-order response. Their mean-square-error (MSE) matrix is
\begin{equation}
 V_{ij}^{(m)}[\Pi,\widehat\theta]
 :=\sum_x\widehat\theta_i(x)\widehat\theta_j(x)\Tr[\rho^{\otimes m}\Pi_x].
\label{eq:mse}
\end{equation}
The exact finite-copy locally unbiased precision value is
\begin{equation}
 C^{(m)}[G]
 :=\inf_{(\Pi,\widehat\theta)\ \mathrm{l.u.}}
 \Tr\!\left[G V^{(m)}[\Pi,\widehat\theta]\right].
\label{eq:Cm-def}
\end{equation}
Equation~\eqref{eq:Cm-def} is the finite-copy precision quantity used throughout the paper. This locally unbiased formulation is standard in quantum estimation theory; see, e.g., Refs.~\cite{Helstrom1976,Holevo2011,Hayashi2005,Matsumoto2002}.

\subsection{Schur--Weyl decomposition of the finite-copy model}
\label{subsec:sw-decomp}

Let $\Ymn$ denote the set of partitions (Young diagrams) $\lambda$ of $m$ with at most $n$ rows,
\begin{equation}
 \Ymn:=\{\lambda\vdash m:\ell(\lambda)\le n\}.
\label{eq:young-index-set}
\end{equation}
Schur--Weyl duality gives~\cite{Sagan2001,Hayashi2017}
\begin{equation}
 \cH^{\otimes m}
 \cong
 \bigoplus_{\lambda\in\Ymn}\cU_\lambda\otimes\cV_\lambda,
\label{eq:general-sw-decomposition}
\end{equation}
where $\cU_\lambda$ carries an irreducible $U(n)$ representation and $\cV_\lambda$ the corresponding irreducible representation of the permutation group $S_m$. We call the direct-sum component indexed by $\lambda$ in Eq.~\eqref{eq:general-sw-decomposition} the $\lambda$-sector. After removing the multiplicity space $\cV_\lambda$, we call $\cU_\lambda$ the reduced block. From the $U(n)$ viewpoint, $\cV_\lambda$ is the multiplicity space: its dimension counts how many equivalent copies of $\cU_\lambda$ occur. Write $f_\lambda:=\dim\cV_\lambda$.

For every $\theta$, the i.i.d. state $\rho_\theta^{\otimes m}$ commutes with all permutations. Hence Schur's lemma gives
\begin{equation}
 \rho_\theta^{\otimes m}
 \cong
 \bigoplus_{\lambda\in\Ymn}R_{\lambda,\theta}\otimes I_{\cV_\lambda}.
\label{eq:block-state}
\end{equation}
Differentiating at $\theta=0$ gives
\begin{equation}
 D_i^{(m)}
 \cong
 \bigoplus_{\lambda\in\Ymn}D_{i,\lambda}\otimes I_{\cV_\lambda}.
\label{eq:block-tangent}
\end{equation}
The multiplicity spaces therefore repeat identical state and tangent blocks.

\subsection{Schur--Weyl reduced model and statistical equivalence}
\label{subsec:reduced-model}

Because the state and tangent operators in Eqs.~\eqref{eq:block-state}--\eqref{eq:block-tangent} act identically on the copies of $\cU_\lambda$ indexed by $\cV_\lambda$, we can represent the parameter-dependent blocks without retaining the multiplicity spaces. To account for all $f_\lambda$ identical copies, their dimensions must enter the reduced state and tangent operators as weights. Starting from the decomposition in Eq.~\eqref{eq:general-sw-decomposition}, define
\begin{equation}
 \cK_m:=\bigoplus_{\lambda\in\Ymn}\cU_\lambda,
\label{eq:Km}
\end{equation}
and the weighted block operators
\begin{equation}
 \bar R_{\lambda,\theta}:=f_\lambda R_{\lambda,\theta},
 \qquad
 \bar D_{i,\lambda}:=f_\lambda D_{i,\lambda}.
\label{eq:weighted-blocks}
\end{equation}
For example, if $H=\bigoplus_\lambda H_\lambda\otimes I_{\cV_\lambda}$ commutes with all permutations, then
\begin{equation}
 \Tr[\rho_\theta^{\otimes m}H]
 =\sum_\lambda f_\lambda\Tr[R_{\lambda,\theta}H_\lambda]
 =\sum_\lambda\Tr[\bar R_{\lambda,\theta}H_\lambda].
\label{eq:weighted-trace}
\end{equation}
Thus the multiplicity weights retain the expectation values of permutation-invariant observables. The preservation of outcome distributions for arbitrary POVMs will follow below from the compression and recovery channels, rather than from Eq.~\eqref{eq:weighted-trace} alone.
On $\cK_m$, define the reduced state family and tangent operators by
\begin{equation}
 \rho_{\theta,\mathrm{red}}^{(m)}
 :=\bigoplus_{\lambda\in\Ymn}\bar R_{\lambda,\theta},
 \qquad
 D_{i,\mathrm{red}}^{(m)}
 :=\bigoplus_{\lambda\in\Ymn}\bar D_{i,\lambda}.
\label{eq:reduced-model}
\end{equation}
We call the family $\{\rho_{\theta,\mathrm{red}}^{(m)}\}_\theta$ on $\cK_m$ the \emph{Schur--Weyl reduced model}. At the operating point we write $\bar R_\lambda:=\bar R_{\lambda,0}$ and $\rho_{\mathrm{red}}^{(m)}:=\rho_{0,\mathrm{red}}^{(m)}$. The quantity $p_\lambda(\theta):=\Tr\bar R_{\lambda,\theta}$ is the probability of the Schur--Weyl sector $\lambda$, with $\sum_\lambda p_\lambda(\theta)=1$.

\Needspace{4\baselineskip}
The concept of sufficient statistics~\cite{Buscemi2012} provides a natural interpretation of the Schur--Weyl reduction. We show that the reduced quantum statistical model preserves all information about the unknown parameters by constructing a compression channel and a parameter-independent recovery channel. Let $P_\lambda^{\mathrm{full}}$ project onto $\cU_\lambda\otimes\cV_\lambda$ in $\cH^{\otimes m}$, and let $P_\lambda^{\mathrm{red}}$ project onto $\cU_\lambda$ in $\cK_m$. Define $\tau_\lambda:=I_{\cV_\lambda}/f_\lambda$. Define the compression and recovery channels by
\begin{align}
 \Phi_m(A)
 &:=\bigoplus_{\lambda\in\Ymn}\Tr_{\cV_\lambda}
   (P_\lambda^{\mathrm{full}} A P_\lambda^{\mathrm{full}}),
 \label{eq:Phi}\\
 \Psi_m(B)
 &:=\bigoplus_{\lambda\in\Ymn}
   (P_\lambda^{\mathrm{red}} B P_\lambda^{\mathrm{red}})\otimes\tau_\lambda.
 \label{eq:Psi}
\end{align}
The channels in Eqs.~\eqref{eq:Phi}--\eqref{eq:Psi} are completely positive and trace preserving (CPTP). Using the decompositions and weights in Eqs.~\eqref{eq:block-state}, \eqref{eq:block-tangent}, and~\eqref{eq:weighted-blocks} gives
\begin{align}
 \Phi_m(\rho_\theta^{\otimes m})&=\rho_{\theta,\mathrm{red}}^{(m)},
 &\Psi_m(\rho_{\theta,\mathrm{red}}^{(m)})&=\rho_\theta^{\otimes m},
 \label{eq:family-recovery}\\
 \Phi_m(D_i^{(m)})&=D_{i,\mathrm{red}}^{(m)},
 &\Psi_m(D_{i,\mathrm{red}}^{(m)})&=D_i^{(m)}.
 \label{eq:tangent-recovery}
\end{align}
Equation~\eqref{eq:family-recovery} shows that the original state family can be exactly recovered from the reduced family for every $\theta$. Therefore, the compression channel $\Phi_m$ is sufficient for the original quantum statistical model. Proposition~\ref{prop:stat-equivalence} establishes the corresponding equivalence of outcome probability distributions under arbitrary POVMs.

\begin{proposition}[Schur--Weyl statistical equivalence]
\label{prop:stat-equivalence}
The full $m$-copy i.i.d. model on $\cH^{\otimes m}$ and the Schur--Weyl reduced model on $\cK_m$ generate the same classical outcome distributions under arbitrary collective POVMs. Consequently, they have the same locally attainable MSE matrices and the same precision value $C^{(m)}[G]$ defined in Eq.~\eqref{eq:Cm-def} for every $G\succ0$.
\end{proposition}

\begin{proof}
Because the channels in Eqs.~\eqref{eq:Phi}--\eqref{eq:Psi} are trace preserving, their adjoints are unital and completely positive and therefore map POVMs to POVMs. If $\{M_x\}$ is a POVM on the full model, define $\widetilde M_x:=\Psi_m^*(M_x)$. Equation~\eqref{eq:family-recovery} gives, for every $\theta$,
\[
 \Tr[\rho_{\theta,\mathrm{red}}^{(m)}\widetilde M_x]
 =\Tr[\rho_\theta^{\otimes m}M_x].
\]
Conversely, for any POVM $\{N_x\}$ on $\cK_m$, the operators $M_x:=\Phi_m^*(N_x)$ form a POVM on the full model and satisfy
\[
 \Tr[\rho_\theta^{\otimes m}M_x]
 =\Tr[\rho_{\theta,\mathrm{red}}^{(m)}N_x]
 \qquad\text{for every }\theta.
\]
Differentiating these equalities at the operating point shows that the derivative responses also agree. Therefore, in either direction we can keep the estimator values unchanged: the two conditions for local unbiasedness in Eqs.~\eqref{eq:local-unbiased-mean}--\eqref{eq:local-unbiased-derivative} and the MSE matrix in Eq.~\eqref{eq:mse} are preserved.
\end{proof}

Figure~\ref{fig:sw-block-diagram} illustrates the removal of permutation multiplicities in a three-copy qubit example.
\begin{figure}[H]
\centering
\includegraphics[width=0.93\linewidth]{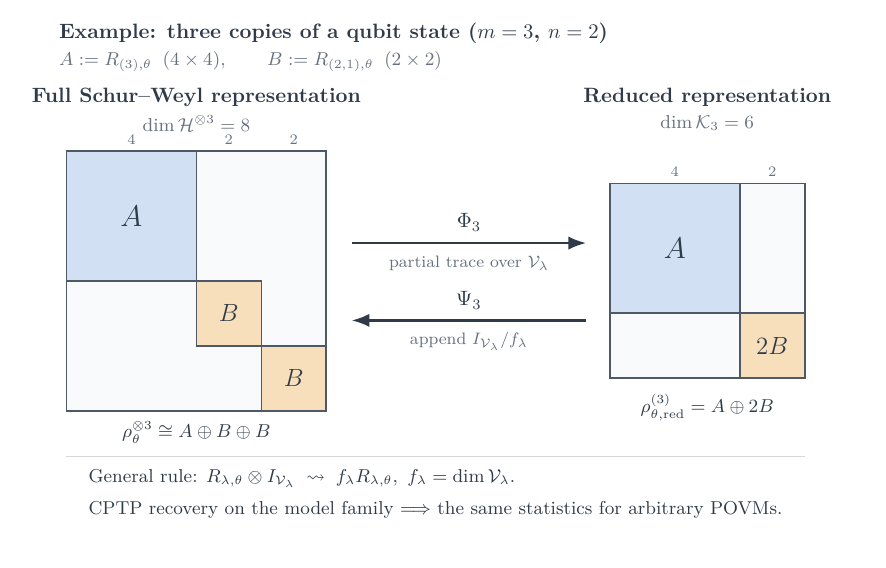}
\caption{Schur--Weyl compression for three identical qubit copies. In a Schur--Weyl basis the full state has blocks $A\oplus B\oplus B$ of sizes $4,2,2$, whereas the reduced state has weighted blocks $A\oplus 2B$ of sizes $4,2$. The channels $\Phi_3$ and $\Psi_3$ recover the two statistical model families from one another; they are not asserted to be inverse maps on arbitrary inputs. The Hilbert space dimensions change from $8$ to $6$, while the dimensions of the corresponding full and permutation-invariant Hermitian operator spaces are $64$ and $20$.}
\label{fig:sw-block-diagram}
\end{figure}

\section{Schur--Weyl reduction of finite-copy precision bounds}
\label{sec:reduced-bounds}

The tight finite-copy precision and its Nagaoka--Hayashi (NH) relaxation have the same objective and statistical constraints in a conic formulation; they differ in the cone to which the optimization variable belongs~\cite{HayashiOuyang2023}. We first formulate the two bounds with a common cone notation. We then apply the Schur--Weyl compression to this formulation and obtain the tight and NH programs as separate specializations. This order separates the symmetry reduction, which is exact for both cones, from the relaxation used only for the NH bound.

\subsection{Conic formulation}
\label{subsec:common-cone}

Introduce the real auxiliary space and the operators that encode the weight and derivative coordinates,
\begin{equation}
 \cR:=\operatorname{span}_{\mathbb R}\{|0\rangle,|1\rangle,\ldots,|d\rangle\}
 \simeq\mathbb R^{d+1},
 \qquad
 \widetilde G:=\sum_{i,j=1}^dG_{ij}|i\rangle\langle j|,
 \qquad
 Q_i:=\frac12\bigl(|0\rangle\langle i|+|i\rangle\langle0|\bigr).
\label{eq:auxiliary-data}
\end{equation}
The auxiliary space $\cR$ allows a POVM and its estimator to be encoded in a single operator. For a POVM $\{\Pi_x\}$ and estimator $\widehat\theta(x)$, define
\[
 |v_x\rangle:=|0\rangle+\sum_{i=1}^d\widehat\theta_i(x)|i\rangle,
 \qquad
 X(\Pi,\widehat\theta):=\sum_x |v_x\rangle\langle v_x|\otimes\Pi_x.
\]
Each term is positive and separable across $\cR$ and $\cH^{\otimes m}$.
For a finite-dimensional Hilbert space $\mathcal A$, let
\begin{equation}
 \cS_{\SEP}(\cR:\mathcal A)
 :=\operatorname{cone}\{A\otimes B:A\in\Sym_+(\cR),\ B\in\Herm_+(\mathcal A)\}
\label{eq:SEP-cone}
\end{equation}
be the real separable cone, and let
\begin{equation}
 \cS_{\mathrm P}(\cR:\mathcal A)
 :=\{X\in\Sym(\cR)\otimes\Herm(\mathcal A): X\succeq0\}
\label{eq:P-cone}
\end{equation}
be the positive-semidefinite cone in the same real-symmetric operator space; positivity in Eq.~\eqref{eq:P-cone} is understood after complexifying the auxiliary factor. These cones satisfy
\[
 \cS_{\SEP}(\cR:\mathcal A)\subseteq\cS_{\mathrm P}(\cR:\mathcal A).
\]
The common optimization objective and constraints below therefore give the NH bound as a lower bound on the tight value.

Using the auxiliary data in Eq.~\eqref{eq:auxiliary-data}, for either cone choice $\mathcal C\in\{\cS_{\SEP},\cS_{\mathrm P}\}$ define the full-space conic program
\begin{equation}
 P_{\mathcal C,\mathrm{full}}^{(m)}[G]
 :=\min_X\Tr[(\widetilde G\otimes\rho^{\otimes m})X]
\label{eq:common-full-primal}
\end{equation}
subject to
\begin{align}
 X&\in\mathcal C(\cR:\cH^{\otimes m}),
 \label{eq:common-full-cone}\\
 \Tr_{\cR}[(|0\rangle\langle0|\otimes I)X]&=I,
 \label{eq:common-full-completeness}\\
 \Tr[(Q_i\otimes D_j^{(m)})X]&=\delta_{ij}.
 \label{eq:common-full-derivative}
\end{align}
The cone constraint in Eq.~\eqref{eq:common-full-cone} selects the precision bound. For a POVM and estimator represented by $X(\Pi,\widehat\theta)$, Eqs.~\eqref{eq:common-full-completeness} and \eqref{eq:common-full-derivative} express POVM completeness and the local derivative condition~\eqref{eq:local-unbiased-derivative}, respectively. The objective in Eq.~\eqref{eq:common-full-primal} equals $\Tr[G V^{(m)}[\Pi,\widehat\theta]]$. The zero-mean condition~\eqref{eq:local-unbiased-mean} is not imposed explicitly, so we must also explain how to obtain it from a feasible separable solution.

First, any feasible separable $X$ admits a finite decomposition into terms $|w_s\rangle\langle w_s|\otimes X_s$, where $X_s\succeq0$. The terms with $w_{s,0}=0$ contribute neither to completeness nor to the derivative conditions. They may be discarded because their contributions to the objective are nonnegative. For each remaining term, define a measurement effect and estimator values by
\begin{equation}
 \Pi_s=|w_{s,0}|^2X_s,\qquad
 \widehat\theta_i(s)=\frac{w_{s,i}}{w_{s,0}}.
\label{eq:sep-povm-reconstruction}
\end{equation}
By Eqs.~\eqref{eq:common-full-completeness} and \eqref{eq:common-full-derivative},
these effects satisfy $\sum_s\Pi_s=I$ and
$\sum_s\widehat\theta_i(s)\Tr(D_j^{(m)}\Pi_s)=\delta_{ij}$.
Thus the reconstructed measurement is normalized and its estimator has
the required local derivatives.

Second, the reconstructed estimator may have a nonzero mean,
$\mu=\sum_s\widehat\theta(s)\Tr(\rho^{\otimes m}\Pi_s)$.
Replacing $\widehat\theta(s)$ by $\widehat\theta(s)-\mu$ sets this mean
to zero without changing the derivative condition, since
$\Tr D_j^{(m)}=0$. The resulting local second-moment matrix obeys
\begin{equation}
 V^{(m)}[\widehat\theta-\mu]
 =V^{(m)}[\widehat\theta]-\mu\mu^T.
\label{eq:sep-mean-centering}
\end{equation}
Since $G\succ0$, centering cannot increase the weighted cost.
The reconstruction therefore supplies a locally unbiased strategy whose
cost does not exceed the separable conic objective.
Appendix~\ref{app:finite-ray} gives the corresponding reconstruction
in sector notation for the finite-ray restriction.

The conic formulation of Hayashi and Ouyang~\cite{HayashiOuyang2023}, extending earlier linear programming approaches to attainable precision bounds~\cite{Hayashi1997}, identifies the full-space program in Eq.~\eqref{eq:common-full-primal} with the two precision quantities as
\begin{equation}
 C^{(m)}[G]=P_{\cS_{\SEP},\mathrm{full}}^{(m)}[G],
 \qquad
 C_{\mathrm{NH,full}}^{(m)}[G]
 =P_{\cS_{\mathrm P},\mathrm{full}}^{(m)}[G].
\label{eq:cone-specializations}
\end{equation}

\subsection{Exact Schur--Weyl reduction}
\label{subsec:common-reduction}

For the Schur--Weyl reduced model defined in Eq.~\eqref{eq:reduced-model}, define the conic program
\begin{equation}
 P_{\mathcal C,\mathrm{red}}^{(m)}[G]
 :=\min_{\{X_\lambda\}}
 \sum_{\lambda\in\Ymn}\Tr[(\widetilde G\otimes\bar R_\lambda)X_\lambda]
\label{eq:common-red-primal}
\end{equation}
subject to
\begin{align}
 X_\lambda&\in\mathcal C(\cR:\cU_\lambda)
 \quad\forall\lambda,
 \label{eq:common-red-cone}\\
 \Tr_{\cR}[(|0\rangle\langle0|\otimes I_{\cU_\lambda})X_\lambda]
 &=I_{\cU_\lambda}
 \quad\forall\lambda,
 \label{eq:common-red-completeness}\\
 \sum_{\lambda\in\Ymn}\Tr[(Q_i\otimes\bar D_{j,\lambda})X_\lambda]
 &=\delta_{ij}.
 \label{eq:common-red-derivative}
\end{align}
Equations~\eqref{eq:common-red-cone}--\eqref{eq:common-red-completeness} are blockwise in the $\lambda$-sectors, whereas Eq.~\eqref{eq:common-red-derivative} remains globally coupled through the sum over $\lambda$. This is the structural form inherited by both the NH and tight specializations.

\begin{theorem}[Exact Schur--Weyl reduction]
\label{thm:common-reduction}
For every $m\ge1$, $G\succ0$, and $\mathcal C\in\{\cS_{\SEP},\cS_{\mathrm P}\}$,
\begin{equation}
 P_{\mathcal C,\mathrm{full}}^{(m)}[G]
 =P_{\mathcal C,\mathrm{red}}^{(m)}[G].
\label{eq:common-reduction-equality}
\end{equation}
\end{theorem}

\begin{proof}
The adjoints $\Phi_m^*$ and $\Psi_m^*$ are unital and completely positive. Acting on the physical factor, each map preserves $\cS_{\mathrm P}$. They also preserve $\cS_{\SEP}$, since each positive-product operator is mapped to another positive-product operator.

Let $X$ satisfy Eqs.~\eqref{eq:common-full-cone}--\eqref{eq:common-full-derivative}. Define
\[
 Y=(\id_{\cR}\otimes\Psi_m^*)(X).
\]
Unitality of $\Psi_m^*$ preserves the completeness constraint. To check the other constraints and the objective, trace duality and Eqs.~\eqref{eq:family-recovery}--\eqref{eq:tangent-recovery} give
\begin{align}
 \Tr[(Q_i\otimes D_{j,\mathrm{red}}^{(m)})Y]
 &=\Tr[(Q_i\otimes D_j^{(m)})X],\nonumber\\
 \Tr[(\widetilde G\otimes\rho_{\mathrm{red}}^{(m)})Y]
 &=\Tr[(\widetilde G\otimes\rho^{\otimes m})X].
\label{eq:primal-transfer-identities}
\end{align}
Thus $Y$ has the same derivative responses and objective as $X$. We then pinch $Y$ across the Schur--Weyl sectors of $\cK_m$. Pinching preserves both cones and the completeness constraint. It also preserves the derivative constraint and objective value, since the reduced state and tangent operators are block diagonal. The resulting operator has the form $\bigoplus_\lambda X_\lambda$ and satisfies Eqs.~\eqref{eq:common-red-cone}--\eqref{eq:common-red-derivative}. Therefore,
\[
 P_{\mathcal C,\mathrm{red}}^{(m)}[G]
 \leq P_{\mathcal C,\mathrm{full}}^{(m)}[G].
\]

Conversely, let $\{X_\lambda\}$ satisfy Eqs.~\eqref{eq:common-red-cone}--\eqref{eq:common-red-derivative}, and set $Y=\bigoplus_\lambda X_\lambda$. Define
\[
 X=(\id_{\cR}\otimes\Phi_m^*)(Y).
\]
The map $\Phi_m^*$ preserves the cone and completeness constraints. The exact state and tangent recovery identities in Eqs.~\eqref{eq:family-recovery}--\eqref{eq:tangent-recovery} also apply in this direction. By trace duality, $X$ satisfies the derivative constraint and has the same objective value as $Y$. Thus, $X$ satisfies Eqs.~\eqref{eq:common-full-cone}--\eqref{eq:common-full-derivative} and has the same objective value, giving
\[
 P_{\mathcal C,\mathrm{full}}^{(m)}[G]
 \leq P_{\mathcal C,\mathrm{red}}^{(m)}[G].
\]
The two inequalities establish Eq.~\eqref{eq:common-reduction-equality}.
\end{proof}

\subsection{Nagaoka--Hayashi specialization}
\label{subsec:nh-specialization}

Choosing $\mathcal C=\cS_{\mathrm P}$ in Theorem~\ref{thm:common-reduction} gives the exact Schur--Weyl reduction of the NH bound. We therefore write
\begin{equation}
 C_{\mathrm{NH,red}}^{(m)}[G]
 :=P_{\cS_{\mathrm P},\mathrm{red}}^{(m)}[G]
 =C_{\mathrm{NH,full}}^{(m)}[G].
\label{eq:NH-equality}
\end{equation}
We express the reduced NH SDP in block-matrix form using the auxiliary-coordinate order $0,1,\ldots,d$ of Eq.~\eqref{eq:auxiliary-data}. For each $\lambda$-sector, define
\[
 Y_\lambda
 :=\begin{pmatrix}
 I_{\cU_\lambda}&\mathbf X_\lambda^T\\
 \mathbf X_\lambda&\mathbb L_\lambda
 \end{pmatrix},
 \qquad
 \mathbf X_\lambda:=(X_{1,\lambda},\ldots,X_{d,\lambda})^T,
 \qquad
 \mathbb L_\lambda=[L_{ij,\lambda}]_{i,j=1}^d.
\]
Here $X_{i,\lambda}\in\Herm(\cU_\lambda)$ represent first estimator moments, while $L_{ij,\lambda}=L_{ji,\lambda}\in\Herm(\cU_\lambda)$ serve as relaxed second moment variables.
The zero mean condition can be imposed without changing the optimal NH value. To see this, let
\[
 \mu_i:=\sum_{\lambda}\Tr[\bar R_\lambda X_{i,\lambda}].
\]
Recenter the variables in every sector by
\begin{align*}
 X_{i,\lambda}'&:=X_{i,\lambda}-\mu_iI_{\cU_\lambda},\\
 L_{ij,\lambda}'&:=L_{ij,\lambda}-\mu_iX_{j,\lambda}
                    -\mu_jX_{i,\lambda}+\mu_i\mu_jI_{\cU_\lambda}.
\end{align*}
With the $0,1,\ldots,d$ block order, this is a congruence by a block-lower-triangular matrix, whose upper-left block is $I_{\cU_\lambda}$ and whose lower-left blocks are $-\mu_i I_{\cU_\lambda}$. Thus $Y_\lambda\succeq0$ implies $Y_\lambda'\succeq0$. The derivative constraints are summed over sectors. In fact,
\begin{align*}
 \sum_\lambda\Tr[\bar D_{j,\lambda}X_{i,\lambda}']
 &=\sum_\lambda\Tr[\bar D_{j,\lambda}X_{i,\lambda}]
   -\mu_i\sum_\lambda\Tr\bar D_{j,\lambda}\\
 &=\delta_{ij},
\end{align*}
because $\sum_\lambda\Tr\bar D_{j,\lambda}=\Tr D_j^{(m)}=0$; the individual sector traces need not vanish. Moreover, $\sum_\lambda\Tr\bar R_\lambda=1$ gives $\sum_\lambda\Tr[\bar R_\lambda X_{i,\lambda}']=0$. Substitution in the objective shows that centering lowers its value by $\mu^TG\mu$. Hence a feasible uncentered point can be replaced by a centered one without increasing the objective, and the NH SDP can be written as follows~\cite{Nagaoka2005,Conlon2021,HayashiOuyang2023}.
\begin{equation}
 C_{\mathrm{NH,red}}^{(m)}[G]
 =\min_{\{X_{i,\lambda},L_{ij,\lambda}\}}
 \sum_{\lambda\in\Ymn}\sum_{i,j=1}^d
 G_{ij}\Tr[\bar R_\lambda L_{ij,\lambda}]
\label{eq:NH-reduced}
\end{equation}
subject to
\begin{align}
 Y_\lambda&\succeq0\quad\forall\lambda,
 \label{eq:NH-red-PSD}\\
 \sum_\lambda\Tr[\bar R_\lambda X_{i,\lambda}]&=0,
 \label{eq:NH-red-centering}\\
 \sum_\lambda\Tr[\bar D_{j,\lambda}X_{i,\lambda}]&=\delta_{ij}.
 \label{eq:NH-red-LU}
\end{align}
This is one SDP, rather than a collection of independently solved block problems: Eq.~\eqref{eq:NH-red-PSD} imposes blockwise positivity, while Eqs.~\eqref{eq:NH-red-centering}--\eqref{eq:NH-red-LU} couple the sectors.

\subsection{Tight bound: reduced primal and dual}
\label{subsec:tight-specialization}

Choosing $\mathcal C=\cS_{\SEP}$ in Eq.~\eqref{eq:common-red-primal} gives the reduced conic program for the exact locally unbiased precision value. We denote its optimal value by $P_{\mathrm{red}}^{(m)}:=P_{\cS_{\SEP},\mathrm{red}}^{(m)}[G]$.

We derive the corresponding dual bound directly from the primal constraints. Choose a real matrix $a=(a_{ij})\in\mathbb R^{d\times d}$ and, for each $\lambda$, a Hermitian operator $\bar S_\lambda\in\Herm(\cU_\lambda)$. Define
\begin{equation}
 \bar\Gamma_\lambda(a,\bar S_\lambda)
 :=\widetilde G\otimes\bar R_\lambda
 -\sum_{i,j=1}^da_{ij}Q_i\otimes\bar D_{j,\lambda}
 -|0\rangle\langle0|\otimes\bar S_\lambda.
\label{eq:bar-Gamma}
\end{equation}
For every family $\{X_\lambda\}$ satisfying the primal equality constraints in Eqs.~\eqref{eq:common-red-completeness}--\eqref{eq:common-red-derivative}, this definition gives
\begin{align}
 &\sum_\lambda\Tr[(\widetilde G\otimes\bar R_\lambda)X_\lambda]
 =\Tr a+\sum_\lambda\Tr\bar S_\lambda\nonumber\\
 &\qquad+\sum_\lambda\Tr[\bar\Gamma_\lambda(a,\bar S_\lambda)X_\lambda].
\label{eq:red-tight-bound-identity}
\end{align}
Indeed, the derivative constraints contribute $\sum_{i,j}a_{ij}\delta_{ij}=\Tr a$, while each completeness constraint contributes $\Tr\bar S_\lambda$. Here $\langle w|Q_i|w\rangle=w_0w_i$ by the definition of $Q_i$ in Eq.~\eqref{eq:auxiliary-data}.

If $\langle w|\bar\Gamma_\lambda(a,\bar S_\lambda)|w\rangle\succeq0$ for every real $w\in\cR$, then $\Tr[\bar\Gamma_\lambda(a,\bar S_\lambda)X_\lambda]\ge0$ for every separable $X_\lambda$: each such $X_\lambda$ is a finite sum of operators $|w\rangle\langle w|\otimes B$ with $B\succeq0$. Consequently, Eq.~\eqref{eq:red-tight-bound-identity} shows that $\Tr a+\sum_\lambda\Tr\bar S_\lambda$ is a lower bound on every feasible primal objective value. Maximizing this lower bound gives the reduced dual program
\begin{equation}
 D_{\mathrm{red}}^{(m)}
 :=\max_{a,\{\bar S_\lambda\}}
 \left(\Tr a+\sum_\lambda\Tr\bar S_\lambda\right)
\label{eq:red-dual}
\end{equation}
subject to
\begin{equation}
 \langle w|\bar\Gamma_\lambda(a,\bar S_\lambda)|w\rangle\succeq0
 \quad\text{for every real }w\in\cR\text{ and every }\lambda.
\label{eq:red-dual-constraint}
\end{equation}
Equation~\eqref{eq:red-dual-constraint} is exactly the block positivity condition for the dual cone $\cS_{\SEP}^*(\cR:\cU_\lambda)$. The derivative constraints involve a sum over all sectors, so the matrix $a$ is shared by them. In contrast, each sector has its own completeness constraint and operator $\bar S_\lambda$.

\begin{corollary}[Primal--dual equality for the reduced tight program]
\label{cor:tight-reduction}
Under the standing assumptions of Sec.~\ref{subsec:local-model},
\begin{equation}
 C^{(m)}[G]=P_{\mathrm{red}}^{(m)}=D_{\mathrm{red}}^{(m)}.
\label{eq:red-equality}
\end{equation}
\end{corollary}

Equation~\eqref{eq:cone-specializations} identifies the full-space separable conic program with the tight precision value $C^{(m)}[G]$. The exact Schur--Weyl reduction in Eq.~\eqref{eq:common-reduction-equality} shows that the reduced primal program has the same optimal value. Equality with the dual optimal value is proved in Appendix~\ref{app:tight-cone}. 

\subsection{Reduction of optimization size}
\label{subsec:size-reduction}

The full Hermitian operator algebra on $\cH^{\otimes m}$ has real dimension $n^{2m}$. By contrast, the multiplicity-free block algebra has dimension
\begin{equation}
 N_{\mathrm{op}}(m,n)
 =\sum_{\lambda\in\Ymn}(\dim\cU_\lambda)^2
 =\binom{m+n^2-1}{n^2-1}.
\label{eq:commutant-dim}
\end{equation}
The blockwise formula in Eq.~\eqref{eq:commutant-dim} establishes the expression for $N_{\mathrm{op}}(m,n)$ announced in Eq.~\eqref{eq:intro-size}. This is the dimension of the permutation commutant in Schur--Weyl duality~\cite{Sagan2001}. The permutation-invariant operator algebra is isomorphic to $\bigoplus_\lambda\End(\cU_\lambda)$. As a vector space, it is also identified with $\Sym^m(\End\cH)$, which gives the dimension formula above.

\begin{corollary}[Copy number scaling]
\label{cor:copy-number-scaling}
For fixed one-copy dimension $n$,
\begin{equation}
 N_{\mathrm{op}}(m,n)=\Theta(m^{n^2-1}),
\end{equation}
whereas the corresponding full operator-algebra dimension is $n^{2m}$.
\end{corollary}

Before Schur--Weyl reduction, the NH SDP imposes a positive-semidefinite constraint on a block matrix of size $(d+1)n^m$. After reduction, each $\lambda$-sector has a positive-semidefinite constraint on a block matrix of size $(d+1)\dim\cU_\lambda$. For a qubit model, $\max_\lambda\dim\cU_\lambda=m+1$, so the largest NH block changes from $(d+1)2^m$ to $(d+1)(m+1)$. In the finite-ray tight program introduced below, each sampled direction introduces one matrix variable on each block $\cU_\lambda$. Equation~\eqref{eq:commutant-dim} therefore gives the total real dimension of these physical matrix variables per sampled direction. The matrix sizes and variable dimensions above describe the optimization formulations. They do not establish bounds on solver runtime.

\section{Quadratic finite-ray approximation of the reduced tight bound}
\label{sec:quadratic}

After Schur--Weyl reduction, the tight dual still requires Eq.~\eqref{eq:red-dual-constraint} for every real direction $w\in\cR$. The continuum of constraints is the remaining numerical difficulty. A finite-ray restriction gives an SDP whose primal value is an upper bound, but the sampled dual constraints alone do not certify a lower bound for the continuous problem. We therefore first construct the finite-ray program, then use the quadratic dependence of its dual constraint on affine estimator coordinates to restore positivity at unsampled rays. The size of this correction is controlled by the mesh covering radius. Throughout this section we take $G=I_d$; the extension to a fixed $G\succ0$ is given in Appendix~\ref{app:general-G}. We use $\|\cdot\|$ for the Euclidean vector norm or operator norm, and $\|\cdot\|_1$ for the trace norm.

\subsection{Finite-ray approximation}
\label{subsec:finite-ray}

Let
\[
 W=(|w_s\rangle)_{s=1}^{N_W},\qquad |w_s\rangle\in\cR,
 \qquad \langle w_s|w_s\rangle=1,
\]
be an indexed family of unit vectors representing sampled projective directions. We use $w\in W$ to mean an entry of this family, so the covering-radius minimum is taken over its indices. Vectors $w$ and $-w$ represent the same rank-one projector $|w\rangle\langle w|$; repeated representatives are permitted and counted in $N_W$. Write $|w_s\rangle=\sum_{i=0}^d w_{s,i}|i\rangle$, so $w_{s,0}=\langle0|w_s\rangle$. Restricting the separable primal variable to
\[
 X_\lambda
 =\sum_{s=1}^{N_W}|w_s\rangle\langle w_s|\otimes X_{s,\lambda},
 \qquad X_{s,\lambda}\succeq0,
\]
gives the finite-ray SDP
\begin{align}
 U_W^{(m)}:=\min_{\{X_{s,\lambda}\}}\;&
 \sum_{\lambda,s}\langle w_s|\widetilde G|w_s\rangle
 \Tr[\bar R_\lambda X_{s,\lambda}],
\label{eq:finite-ray-primal}\\
\text{subject to }\;&
 \sum_s|w_{s,0}|^2X_{s,\lambda}=I_{\cU_\lambda}
 \quad\forall\lambda,
\label{eq:finite-completeness}\\
&\sum_{\lambda,s}\langle w_s|Q_i|w_s\rangle
 \Tr[\bar D_{j,\lambda}X_{s,\lambda}]=\delta_{ij}
 \quad i,j=1,\ldots,d.
\label{eq:finite-LU}
\end{align}
Equations~\eqref{eq:finite-ray-primal}--\eqref{eq:finite-LU} are the finite-ray counterparts of the reduced tight constraints in Sec.~\ref{subsec:tight-specialization}. Because $W$ contains finitely many directions, the resulting problem has finitely many matrix variables $X_{s,\lambda}\succeq0$. Its objective, completeness constraints, and derivative constraints are linear in the matrices $X_{s,\lambda}$. Hence, the finite-ray problem is an ordinary SDP. Not every choice of $W$ admits positive-semidefinite operators $X_{s,\lambda}$ satisfying Eqs.~\eqref{eq:finite-completeness} and~\eqref{eq:finite-LU}. Whenever $U_W^{(m)}$ is used below, such operators are assumed to exist. Any such solution reconstructs a locally unbiased measurement strategy after recentering the estimator values, so
\begin{equation}
 C^{(m)}[I_d]\le U_W^{(m)}.
\label{eq:upper-bound}
\end{equation}
The reconstruction, including rays with $w_{s,0}=0$, is given in Appendix~\ref{app:finite-ray}.

The dual of this finite SDP relaxes the continuous constraint~\eqref{eq:red-dual-constraint} by requiring positivity only on the sampled rays:
\begin{equation}
 \langle w_s|\bar\Gamma_\lambda(a,\bar S_\lambda)|w_s\rangle\succeq0,
 \qquad s=1,\ldots,N_W.
\label{eq:finite-ray-dual-rays}
\end{equation}
The sampled dual objective is not in general a lower bound on the original continuous tight problem: its constraints have been checked only at the rays in $W$. For any $(a,\bar S)$ satisfying the sampled dual constraints in Eq.~\eqref{eq:finite-ray-dual-rays}, define
\begin{equation}
 J_W^{(m)}(a,\bar S)
 :=\Tr a+\sum_\lambda\Tr\bar S_\lambda,
 \qquad
 \varepsilon_W:=U_W^{(m)}-J_W^{(m)}(a,\bar S)\ge0.
\label{eq:finite-ray-dual-gap}
\end{equation}
When the finite-ray primal constraints admit a solution, the primal and dual have the same optimal value $U_W^{(m)}$; Appendix~\ref{app:finite-ray} records the finite-dimensional strong-duality argument. Consequently, for every $\varepsilon>0$ there exist $(a,\bar S)$ satisfying Eq.~\eqref{eq:finite-ray-dual-rays} with $\varepsilon_W<\varepsilon$. Such pairs may violate positivity at unsampled rays. Figure~\ref{fig:projective-mesh-illustration} illustrates the distinction between the continuous family of projective directions and its finite sampling. The next step is to restore positivity for all projective directions.

\begin{figure}[H]
\centering
\includegraphics[width=0.88\linewidth]{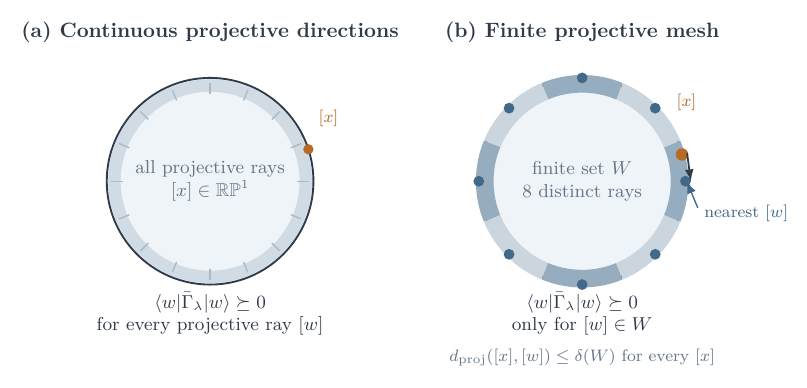}
\caption{Sampling a continuum of projective directions with a finite family of rays. For clarity the drawing uses the one-dimensional real projective space $\mathbb{RP}^1$, represented by a circle (antipodal vectors identify the same ray). The shaded arcs around the eight distinct sampled rays cover the entire circle. In general dimension, the covering radius used in our bounds is the projective trace-norm distance $\delta(W)$ in Eq.~\eqref{eq:projective-radius}, not a Euclidean radius in affine coordinates. Projective coverage alone does not ensure that the dual constraint holds at unsampled directions.}
\label{fig:projective-mesh-illustration}
\end{figure}

\subsection{Quadratic structure of the dual constraint}
\label{subsec:quadratic-geometry}

For a ray with $w_0\neq0$, define the affine coordinates $y_i=w_i/w_0$. These ratios give the estimator values reconstructed in Appendix~\ref{app:finite-ray}. Fix $(a,\{\bar S_\lambda\})$ satisfying the sampled dual constraints in Eq.~\eqref{eq:finite-ray-dual-rays}. Because $\rho\succ0$, Eq.~\eqref{eq:block-state} implies $\bar R_\lambda\succ0$ on every reduced block, so the following inverse square roots are defined:
\[
 A_{i,\lambda}:=\sum_{j=1}^da_{ij}\bar D_{j,\lambda},
 \qquad
 B_{i,\lambda}:=\bar R_\lambda^{-1/2}A_{i,\lambda}\bar R_\lambda^{-1/2},
 \qquad
 \widetilde S_\lambda:=\bar R_\lambda^{-1/2}\bar S_\lambda\bar R_\lambda^{-1/2}.
\]
The continuous dual constraint in Eq.~\eqref{eq:red-dual-constraint} requires the operator $\langle w|\bar\Gamma_\lambda(a,\bar S_\lambda)|w\rangle$ to be positive semidefinite for every real $w$ and every sector $\lambda$. For $w_0\neq0$, we divide this operator by $w_0^2$. This gives
\[
 \frac{\langle w|\bar\Gamma_\lambda(a,\bar S_\lambda)|w\rangle}{w_0^2}
 =\|y\|^2\bar R_\lambda-\sum_{i=1}^d y_i A_{i,\lambda}-\bar S_\lambda.
\]
Applying a congruence transformation with $\bar R_\lambda^{-1/2}$ changes the coefficient of $\|y\|^2$ from $\bar R_\lambda$ to the identity operator. We obtain
\begin{equation}
 \widetilde M_\lambda(y)
 :=\|y\|^2I-\sum_{i=1}^dy_iB_{i,\lambda}-\widetilde S_\lambda.
\label{eq:Mtilde}
\end{equation}
The full continuous dual requires $\widetilde M_\lambda(y)\succeq0$ for every $y\in\mathbb R^d$, whereas the finite-ray dual checks this condition only at the affine points induced by the sampled rays.
For a unit vector $u\in\cU_\lambda$, define
\[
 b_{\lambda,u}
 :=\bigl(\langle u|B_{1,\lambda}|u\rangle,\ldots,\langle u|B_{d,\lambda}|u\rangle\bigr)^T,
 \qquad
 y_{\lambda,u}:=\frac12 b_{\lambda,u}.
\]
\Needspace{7\baselineskip}
Consider the scalar function
\[
 q_{\lambda,u}(y):=\langle u|\widetilde M_\lambda(y)|u\rangle.
\]
Completing the square shows that $y_{\lambda,u}$ is its unique minimizer. Moreover,
\begin{equation}
 q_{\lambda,u}(z)-q_{\lambda,u}(y_{\lambda,u})
 =\|z-y_{\lambda,u}\|^2.
\label{eq:quadratic-identity}
\end{equation}
Equation~\eqref{eq:quadratic-identity} is exact and is the source of the quadratic error scaling; no Taylor expansion is involved. Figure~\ref{fig:quadratic-repair-geom} illustrates the exact quadratic identity in one affine dimension.

\begin{figure}[H]
\centering
\includegraphics[width=0.82\linewidth]{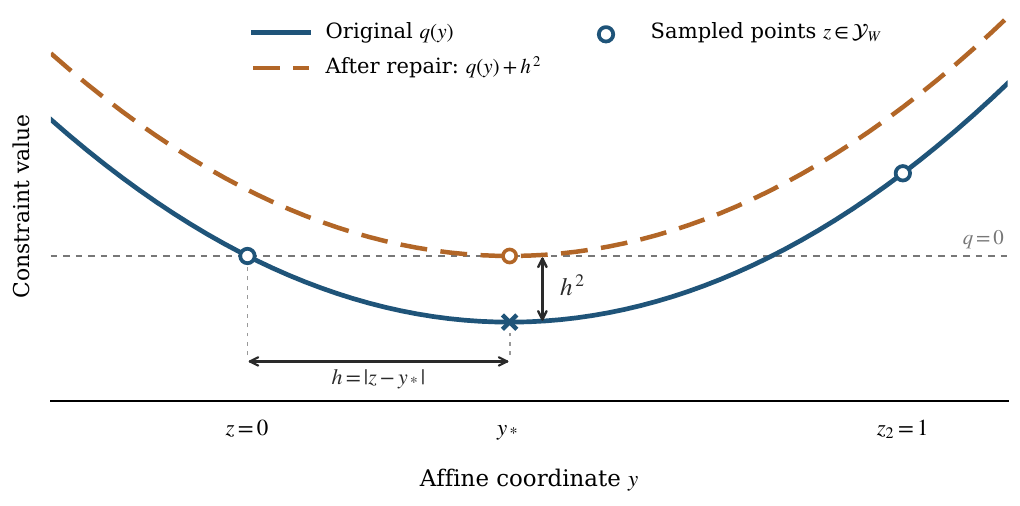}
\caption{Quadratic dual repair for a fixed sector $\lambda$, unit vector $u$, and one affine coordinate $y$. The illustrative function $q(y)=(y-0.4)^2-0.16$ is nonnegative at the sampled points $z=0,1$ but has minimum $-0.16$ between them. With $h=|0-0.4|=0.4$, the exact identity $q(z)-q(y_*)=h^2$ gives the conservative upward shift $q(y)+h^2\ge0$. The full matrix-valued repair uses a uniform $h_\lambda$ over all unit vectors $u$; its relation to the projective mesh radius follows from Eq.~\eqref{eq:eta} under the stated conditions.}
\label{fig:quadratic-repair-geom}
\end{figure}

\subsection{Dual repair and two-sided bounds}
\label{subsec:dual-repair}

Let
\[
 \mathcal Y_W:=\left\{\left(\frac{w_{s,1}}{w_{s,0}},\ldots,\frac{w_{s,d}}{w_{s,0}}\right): w_s\in W,\ w_{s,0}\ne0\right\}
\]
be the affine mesh induced by rays with nonzero zeroth component, and define
\[
 K_\lambda(a):=\{y_{\lambda,u}:\|u\|=1\},
 \qquad
 h_\lambda^2
 :=\sup_{y\in K_\lambda(a)}\min_{z\in\mathcal Y_W}\|z-y\|^2.
\]
Thus $h_\lambda$ is the worst-case Euclidean distance from a sector minimizer to the sampled affine mesh. Fix a unit vector $u$, and choose a nearest sampled point $z\in\mathcal Y_W$ to $y_{\lambda,u}$. The sampled constraint in Eq.~\eqref{eq:finite-ray-dual-rays} gives $q_{\lambda,u}(z)\ge0$. By the exact identity~\eqref{eq:quadratic-identity},
\begin{equation}
 q_{\lambda,u}(y_{\lambda,u})
 =q_{\lambda,u}(z)-\|z-y_{\lambda,u}\|^2
 \ge-h_\lambda^2.
\label{eq:sector-repair-minimum}
\end{equation}
Because $y_{\lambda,u}$ minimizes $q_{\lambda,u}$, this inequality holds with $q_{\lambda,u}(y)$ on the left for every $y$. Varying $u$ therefore gives
\[
 \widetilde M_\lambda(y)\succeq-h_\lambda^2I
 \qquad\forall y\in\mathbb R^d.
\]
Hence
\begin{equation}
 \bar S_\lambda'
 :=\bar S_\lambda-h_\lambda^2\bar R_\lambda
\label{eq:sector-repair}
\end{equation}
satisfies the full continuous dual constraint in Eq.~\eqref{eq:red-dual-constraint}. This subtraction compensates for the deficit allowed by the mesh estimate. The affine argument covers directions with $w_0\neq0$. For $w_0=0$, the original dual constraint reduces to $\langle w|\widetilde G|w\rangle\bar R_\lambda\succeq0$ and does not depend on $\bar S_\lambda$. Thus the repaired variables satisfy the dual constraint at every projective direction. If $p_\lambda:=\Tr\bar R_\lambda$, then $\sum_\lambda p_\lambda=1$ and
\begin{equation}
 J_W^{(m)}(a,\bar S)-\sum_\lambda p_\lambda h_\lambda^2
 \le C^{(m)}[I_d]\le U_W^{(m)}.
\label{eq:sector-two-sided}
\end{equation}
A global bound follows by controlling all sector minimizers with the projective mesh radius.

\subsection{Quadratic accuracy and mesh complexity}
\label{subsec:mesh-complexity}

Define the projective covering radius, which measures the worst-case distance from any projective direction to its nearest sampled ray,
\begin{equation}
 \delta(W)
 :=\max_{\|x\|=1}\min_{w\in W}
 \bigl\||x\rangle\langle x|-|w\rangle\langle w|\bigr\|_1.
\label{eq:projective-radius}
\end{equation}
Because the radius in Eq.~\eqref{eq:projective-radius} is written in terms of rank-one projectors, it is unchanged by $x\mapsto-x$ or $w\mapsto-w$. For $C\ge0$ and $0\le\delta<2/\sqrt{1+C^2}$, define
\begin{equation}
 \eta(C,\delta)
 :=\frac{(1+C^2)\delta}{2\sqrt{1-\delta^2/4}-C\delta}.
\label{eq:eta}
\end{equation}
The restriction ensures that the denominator is positive. For fixed $C$ and $\delta\to0$, we have $\eta(C,\delta)=\frac{1+C^2}{2}\delta+O(\delta^2)$. Lemma~\ref{lem:projective-affine} shows that $\eta(C,\delta)$ bounds the affine-coordinate error for points of norm at most $C$ when the projective mesh has radius at most $\delta$. Let $(a,\{\bar S_\lambda\})$ satisfy the sampled dual constraints in Eq.~\eqref{eq:finite-ray-dual-rays}. Define the full-space normalized operators
\[
 B_i^{\mathrm{full}}(a)
 :=(\rho^{\otimes m})^{-1/2}
 \left(\sum_{j=1}^da_{ij}D_j^{(m)}\right)
 (\rho^{\otimes m})^{-1/2}.
\]
The Hayashi--Ouyang compact-search radius~\cite{HayashiOuyang2023} is
\begin{equation}
 C_2^{\mathrm{HO}}(a)
 :=\frac12\left\|\sum_{i=1}^d\bigl(B_i^{\mathrm{full}}(a)\bigr)^2\right\|^{1/2}.
\label{eq:C2-HO}
\end{equation}
The radius $C_2^{\mathrm{HO}}(a)$ bounds the norms of the affine minimizers $y_{\lambda,u}$. The normalized operators $B_i^{\mathrm{full}}(a)$ decompose into Schur--Weyl blocks:
\[
 B_i^{\mathrm{full}}(a)\cong\bigoplus_\lambda B_{i,\lambda}\otimes I_{\cV_\lambda}.
\]
This identity is proved in Eq.~\eqref{eq:B-full-block} of Appendix~\ref{app:identity-repair}, using the state and tangent block decompositions in Eqs.~\eqref{eq:block-state}--\eqref{eq:block-tangent} and the multiplicity weights in Eq.~\eqref{eq:weighted-blocks}. Squaring and summing these block-diagonal operators gives an explicit reduced expression for the compact-search radius:
\begin{equation}
 C_2^{\mathrm{HO}}(a)
 =\frac12\max_{\lambda\in\Ymn}
 \left\|\sum_{i=1}^d B_{i,\lambda}^2\right\|^{1/2}.
\label{eq:C2-HO-reduced}
\end{equation}
Thus $C_2^{\mathrm{HO}}(a)$ can be evaluated on the reduced blocks without constructing the full $m$-copy matrices. Applying the Cauchy--Schwarz inequality, as shown in the same appendix, gives $\|b_{\lambda,u}\|\le 2C_2^{\mathrm{HO}}(a)$. Since $y_{\lambda,u}=\frac12 b_{\lambda,u}$, we obtain
\begin{equation}
 \|y_{\lambda,u}\|\le C_2^{\mathrm{HO}}(a).
\label{eq:identity-radius}
\end{equation}

\begin{theorem}[Quadratic finite-ray bounds]
\label{thm:global-mesh-bound}
Let $W$ be a mesh for which the finite-ray primal constraints admit a solution. Let $(a,\{\bar S_\lambda\})$ satisfy the sampled dual constraints in Eq.~\eqref{eq:finite-ray-dual-rays} for the identity-weight problem, and define $\varepsilon_W$ by Eq.~\eqref{eq:finite-ray-dual-gap}. If
\[
 \delta(W)<\frac{2}{\sqrt{1+C_2^{\mathrm{HO}}(a)^2}},
\]
then
\begin{equation}
 0\le U_W^{(m)}-C^{(m)}[I_d]
 \le
 \varepsilon_W+\eta\!\left(C_2^{\mathrm{HO}}(a),\delta(W)\right)^2.
\label{eq:global-error-bound}
\end{equation}
In particular,
\begin{equation}
 L_W^{(m)}(a,\bar S)
 :=J_W^{(m)}(a,\bar S)-\eta\!\left(C_2^{\mathrm{HO}}(a),\delta(W)\right)^2
\label{eq:L-global}
\end{equation}
satisfies
\begin{equation}
 L_W^{(m)}(a,\bar S)\le C^{(m)}[I_d]\le U_W^{(m)},
 \qquad
 U_W^{(m)}-L_W^{(m)}(a,\bar S)
 =\varepsilon_W+\eta\!\left(C_2^{\mathrm{HO}}(a),\delta(W)\right)^2.
\label{eq:global-enclosure-width}
\end{equation}
\end{theorem}

\begin{proof}
By Eq.~\eqref{eq:identity-radius}, every sector minimizer lies in the ball of radius $C_2^{\mathrm{HO}}(a)$.

Lemma~\ref{lem:projective-affine} therefore gives
\[
 h_\lambda^2\le\eta\!\left(C_2^{\mathrm{HO}}(a),\delta(W)\right)^2
 \qquad\text{for every }\lambda.
\]
Substituting this bound into Eq.~\eqref{eq:sector-two-sided} and using $\sum_\lambda p_\lambda=1$ gives the lower bound in Eq.~\eqref{eq:L-global}. Since $J_W^{(m)}(a,\bar S)=U_W^{(m)}-\varepsilon_W$, Eqs.~\eqref{eq:global-error-bound} and~\eqref{eq:global-enclosure-width} follow.
\end{proof}

\begin{corollary}[Quadratic accuracy on sufficiently fine meshes]
\label{cor:second-order}
Fix the copy number $m$ and parameter dimension $d$ under the standing assumptions of Sec.~\ref{subsec:local-model}. Let $\{W_k\}$ be any sequence of finite indexed projective meshes with $\delta(W_k)\to0$. For all sufficiently large $k$, the finite-ray primal is feasible. Moreover, sampled dual-feasible pairs $(a_k,\bar S_k)$ can be chosen with
\[
 0\le U_{W_k}^{(m)}-J_{W_k}^{(m)}(a_k,\bar S_k)\le\delta(W_k)^2,
 \qquad \sup_k C_2^{\mathrm{HO}}(a_k)<\infty,
\]
where the supremum is taken over sufficiently large $k$. For these choices,
\begin{equation}
 U_{W_k}^{(m)}-C^{(m)}[I_d]=O(\delta(W_k)^2),
 \qquad
 U_{W_k}^{(m)}-L_{W_k}^{(m)}(a_k,\bar S_k)=O(\delta(W_k)^2).
\label{eq:unconditional-quadratic}
\end{equation}
\end{corollary}
\begin{proof}
Lemma~\ref{lem:fine-mesh-feasible} proves eventual primal feasibility, and finite-SDP strong duality from Appendix~\ref{app:finite-ray} then permits a sampled dual gap at most $\delta(W_k)^2$. Lemma~\ref{lem:uniform-dual} gives the stated uniform bound on $C_2^{\mathrm{HO}}(a_k)$. The strict covering-radius condition of Theorem~\ref{thm:global-mesh-bound} holds for sufficiently large $k$. Applying that theorem proves both estimates in Eq.~\eqref{eq:unconditional-quadratic}.
\end{proof}

For the indexed reference meshes $W_{L,d}$ used in the numerical calculations, whose recursive construction is given in Appendix~\ref{app:projective-affine},
\begin{equation}
 N_W=L^d,
 \qquad
 \delta(W_{L,d})\le\bar\delta_{L,d}:=2\sqrt{1-\cos^{2d}(\pi/L)}
 =\frac{2\pi\sqrt d}{L}+O(L^{-3}).
\label{eq:reference-mesh-scaling}
\end{equation}

\begin{corollary}[Quadratic convergence for the reference meshes]
\label{cor:reference-quadratic}
For fixed $m,d$, the finite-ray primal on $W_{L,d}$ is feasible for all sufficiently large even $L$. With sampled dual-feasible pairs chosen as in Corollary~\ref{cor:second-order},
\[
 0\le U_{W_{L,d}}^{(m)}-C^{(m)}[I_d]=O(L^{-2}),
 \qquad
 U_{W_{L,d}}^{(m)}-L_{W_{L,d}}^{(m)}(a_L,\bar S_L)
 =O(L^{-2})=O(N_W^{-2/d}).
\]
Consequently, for any sufficiently small target width $\epsilon>0$, $O(\epsilon^{-d/2})$ sampled rays suffice to make the theoretical upper--lower enclosure at most $\epsilon$.
\end{corollary}

The feasibility and boundedness results underlying these corollaries are proved in Appendix~\ref{app:uniform-mesh}. Their constants depend on the fixed statistical model and copy number $m$. The results do not assert uniformity as $m$ increases.

\section{Numerical results}
\label{sec:numerics}

We use the two-parameter qubit local model
\begin{equation}
 \rho_r=\frac12(I+r\sigma_z),
 \qquad
 D_1=\frac12\sigma_z,
 \qquad
 D_2=\frac12\sigma_x,
 \qquad
 G=I_2,
\label{eq:numerical-qubit-model}
\end{equation}
with $r\in\{0,0.5\}$, where $\sigma_x$ and $\sigma_z$ are the Pauli matrices and $r$ is the $z$-component of the Bloch vector. We use this model throughout the numerical section. Our main focus is the tight precision bound: the quadratic mesh correction and the finite-ray dual values across copy numbers. We first present the NH calculations to validate the Schur--Weyl reduction and establish a numerical benchmark for the tight calculations that follow. Table~\ref{tab:status-numeric-quantities} distinguishes mathematical bounds from evaluated or solver-reported quantities.

\begin{table}[H]
\centering
\small
\caption{Mathematical and numerical quantities used in the finite-copy comparisons.}
\label{tab:status-numeric-quantities}
\begin{tabular}{@{}p{0.24\linewidth}p{0.67\linewidth}@{}}
\toprule
Quantity & Mathematical role and condition \\
\midrule
$C^{(m)}[I_d]$, $C_{\mathrm{NH}}^{(m)}$ & Exact tight optimum and exact NH relaxation optimum, respectively; $C_{\mathrm{NH}}^{(m)}\le C^{(m)}[I_d]$. Plotted NH values approximate the exact optimum. \\
$U_W^{(m)}$ & Exact finite-ray primal optimum, defined when the sampled primal is feasible; an upper bound on $C^{(m)}[I_d]$. \\
$J_W^{(m)}(a,\bar S)$, $\varepsilon_W$ & Sampled dual objective and finite-ray primal--dual gap for an exactly sampled-dual-feasible pair; $J_W$ alone need not be a lower bound on the continuous problem. \\
$L_W^{(m)}(a,\bar S)$ & Rigorous lower bound on $C^{(m)}[I_d]$ when sampled dual feasibility and the covering-radius condition of Theorem~\ref{thm:global-mesh-bound} hold. \\
$R_W^{(m)}$ & Correction formula evaluated using the upper radius bound $\bar\delta_{L,2}$. \\
$\widehat D_W^{(m)}$ & Numerical dual objective returned for a fixed mesh. \\
\bottomrule
\end{tabular}
\end{table}

\subsection{NH validation and large-copy benchmark}
\label{subsec:nh-numerics}

For the qubit model in Eq.~\eqref{eq:numerical-qubit-model}, write $C_{\mathrm{NH}}^{(m)}:=C_{\mathrm{NH,full}}^{(m)}[I_2]=C_{\mathrm{NH,red}}^{(m)}[I_2]$. The full-space NH positive-semidefinite block matrix has size $3\cdot2^m$. For qubits, the Schur--Weyl sectors may be labeled by total spin $j$, and the corresponding reduced block has dimension $2j+1$~\cite{Sagan2001}. Hence the reduced NH constraints have sizes $3(2j+1)$ and largest block $3(m+1)$. We first compare the two formulations for $m=1,\ldots,5$. Since the leading asymptotic MSE scales as $1/m$, we multiply the NH cost by $m$ and report $mC_{\mathrm{NH}}^{(m)}$. This quantity can be compared with the asymptotic Holevo benchmark for the same local model. Table~\ref{tab:NH-validation} shows agreement between the full and reduced formulations to the numerical precision achieved by the solvers.

\begin{table}[H]
\centering
\small
\caption{Small-copy numerical comparison of the full and Schur--Weyl-reduced NH formulations. Entries are numerical approximations to the normalized optima $mC_{\mathrm{NH}}^{(m)}$.}
\label{tab:NH-validation}
\begin{tabular}{c@{\qquad}cc@{\qquad}cc}
\toprule
& \multicolumn{2}{c}{$r=0$} & \multicolumn{2}{c}{$r=0.5$}\\
$m$ & full & reduced & full & reduced\\
\midrule
1 & 3.999997 & 3.999997 & 3.482051 & 3.482051\\
2 & 2.999976 & 2.999986 & 2.516091 & 2.516075\\
3 & 2.835505 & 2.835505 & 2.328549 & 2.328545\\
4 & 2.684747 & 2.684746 & 2.192318 & 2.192273\\
5 & 2.619359 & 2.619378 & 2.121033 & 2.121022\\
\bottomrule
\end{tabular}
\end{table}

The largest absolute difference between the displayed full and reduced values in Table~\ref{tab:NH-validation} is approximately $4.5\times10^{-5}$ (at $m=4$, $r=0.5$). The exact equality of the formulations follows from Eq.~\eqref{eq:NH-equality}.

The reduced NH SDP can be solved for much larger copy numbers than the full-space NH SDP. Figure~\ref{fig:NH-large-copy} shows the normalized reduced NH values through $m=100$ for $r=0$ and through $m=25$ for $r=0.5$. For this benchmark, the symmetric logarithmic derivative (SLD) Fisher information matrix is $J^{\mathrm{SLD}}(r)=\operatorname{diag}((1-r^2)^{-1},1)$. The mean SLD commutator vanishes, so the model satisfies the weak commutativity condition. Therefore, the Holevo bound is $C_H(r)=\Tr[(J^{\mathrm{SLD}}(r))^{-1}]=2-r^2$~\cite{Holevo2011,Demkowicz2020}. The computed finite-copy NH values decrease toward this asymptotic benchmark over the displayed copy numbers.

\begin{figure}[H]
\centering
\includegraphics[width=0.78\textwidth]{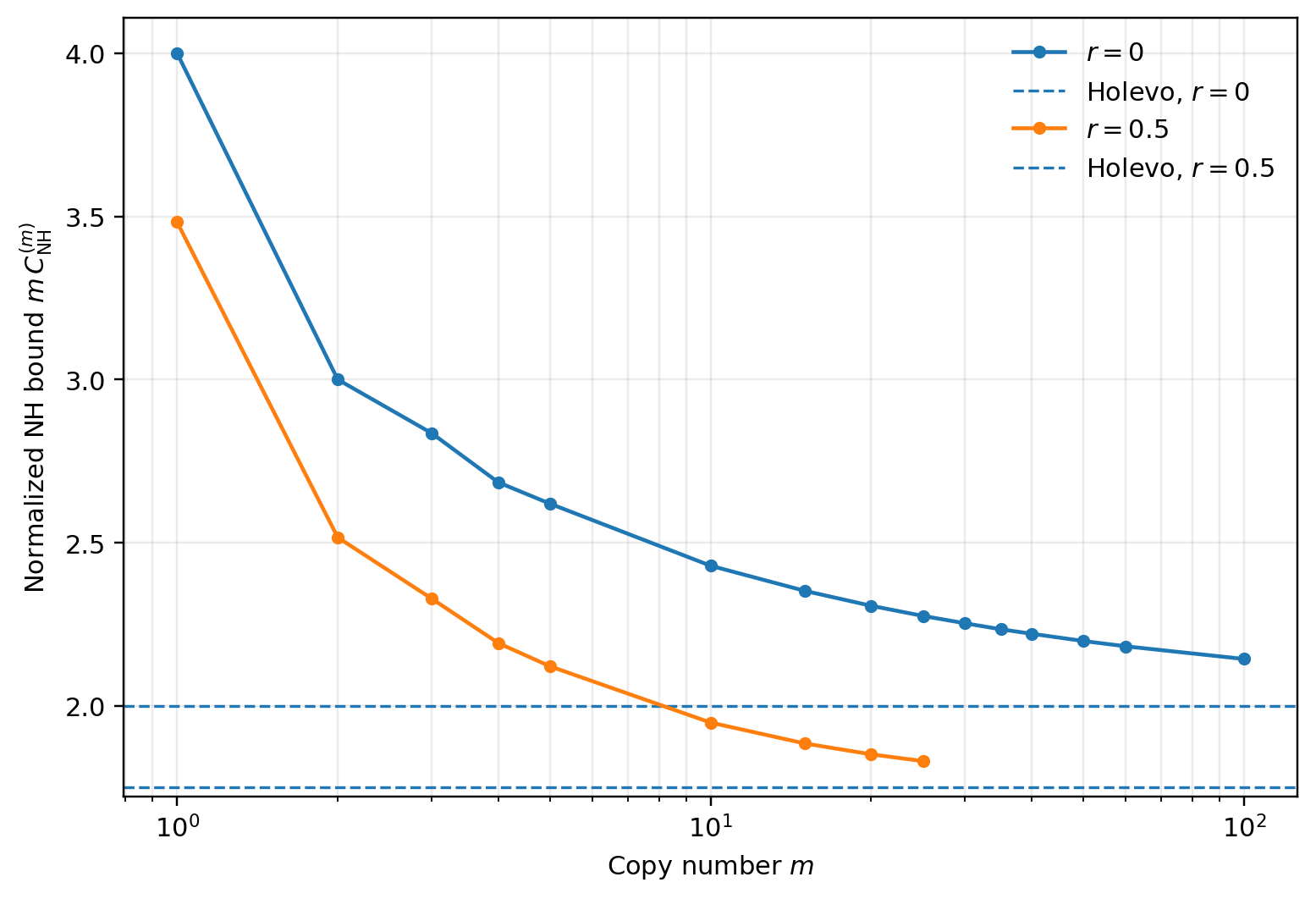}
\caption{NH values multiplied by the copy number $m$ for the Schur--Weyl-reduced qubit SDP. The dashed horizontal levels are the Holevo values $C_H(0)=2$ and $C_H(0.5)=1.75$.}
\label{fig:NH-large-copy}
\end{figure}

The calculation illustrates the practical change in the accessible finite-copy regime after removing permutation multiplicities. At $m=100$, for example, the largest reduced NH block has size $303$, while the corresponding full-space NH block-matrix size is $3\cdot2^{100}$.

\subsection{Quadratic mesh correction for the tight bound}
\label{subsec:tight-mesh-numerics}

For the tight program we use the same model in Eq.~\eqref{eq:numerical-qubit-model} and numerically solve the finite-ray dual for $m=2,3$ and $r=0,0.5$. Let $a_W$ denote the multiplier matrix returned by the finite-ray dual solve on mesh $W$. We evaluate the quadratic repair penalty
\[
 R_W^{(m)}:=\eta\!\left(C_2^{\mathrm{HO}}(a_W),\bar\delta_{L,2}\right)^2.
\]
The function $\eta(C,\delta)$ in Eq.~\eqref{eq:eta} converts a projective covering-radius bound into an affine-coordinate error bound, and $C_2^{\mathrm{HO}}(a)$ in Eq.~\eqref{eq:C2-HO} controls the norms of the quadratic minimizers. Appendix~\ref{app:projective-affine} gives the explicit upper bound $\bar\delta_{L,2}\ge\delta(W)$. Under the smallness condition of Theorem~\ref{thm:global-mesh-bound}, the difference between the exact finite-ray primal upper bound and the repaired dual lower bound is $\varepsilon_W+\eta(C_2^{\mathrm{HO}}(a_W),\delta(W))^2$ for a feasible primal and a sampled-dual-feasible pair. Substituting $\bar\delta_{L,2}$ yields $R_W^{(m)}$ as a conservative correction when the theorem's strict smallness condition is satisfied.

The horizontal axis uses $1/\bar\delta_{L,2}$ so that refinement proceeds from left to right. With $N_W=L^2$, the displayed resolutions are $L\in\{36,44,52,64,80,96,112,128,144,160\}$ for $m=2$ and $L\in\{36,40,44,48,52,56,64,72,80\}$ for $m=3$.

\begin{figure}[H]
\centering
\includegraphics[width=0.82\textwidth]{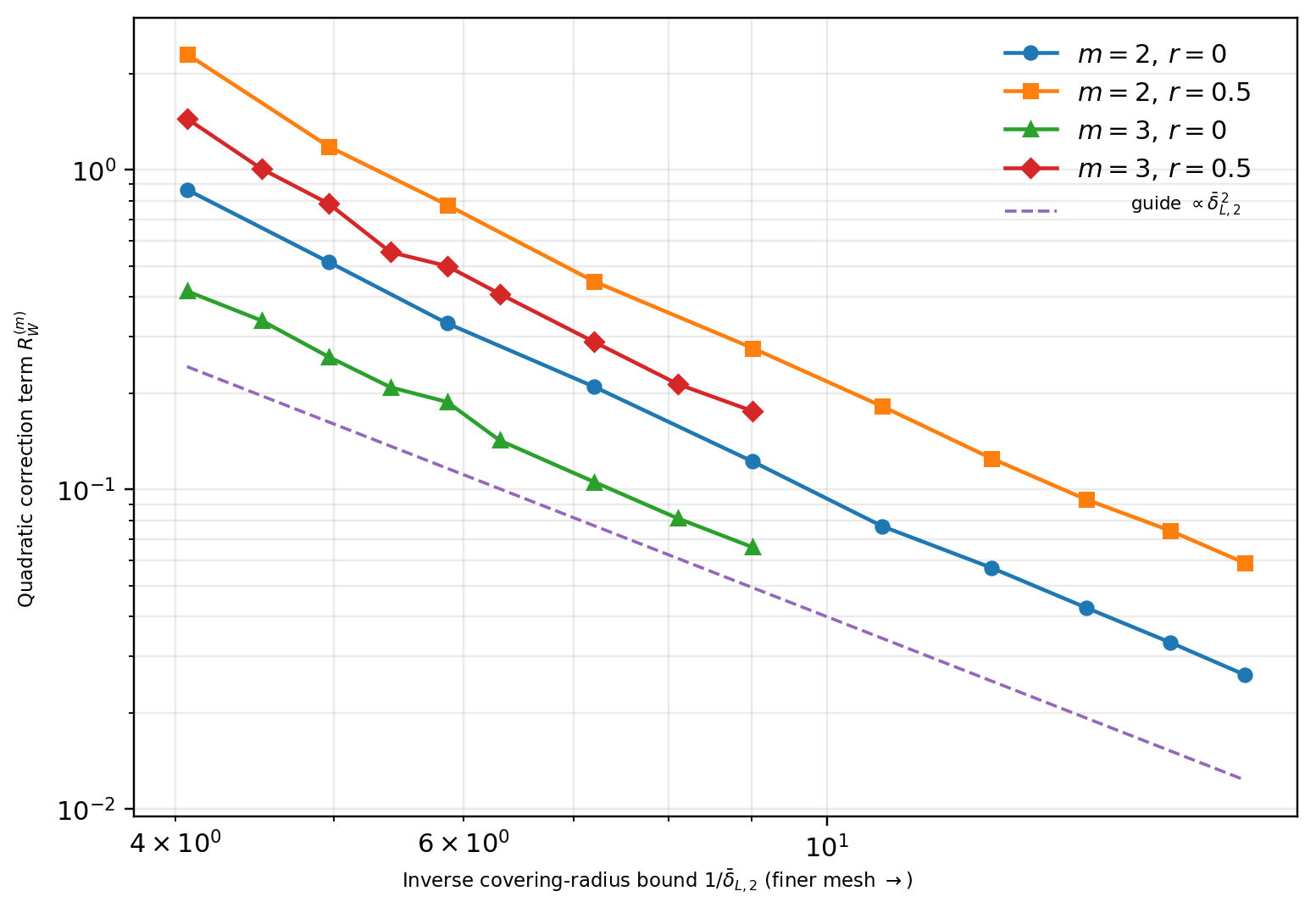}
\caption{Quadratic mesh-correction term $R_W^{(m)}=\eta(C_2^{\mathrm{HO}}(a_W),\bar\delta_{L,2})^2$ for two and three copies.}
\label{fig:quadratic-scaling}
\end{figure}

At the finest displayed meshes, the ratio $R_W^{(m)}/\bar\delta_{L,2}^2$ is approximately $8.53$ at $r=0$ and $19.09$ at $r=0.5$ when $m=2$. For $m=3$, the ratios are approximately $5.35$ at $r=0$ and $14.27$ at $r=0.5$. These ratios change only modestly under further refinement, consistent with the quadratic scaling in Corollary~\ref{cor:second-order}.

\subsection{Fixed-mesh tight computation through ten copies}
\label{subsec:tight-copy-hierarchy}

We next keep the two-parameter reference mesh fixed at resolution $L=48$. Write $W_{48}:=W_{48,2}$ for this mesh; it has $N_{W_{48}}=2304$ sampled vectors and $\bar\delta_{48,2}\simeq0.18479$ as a covering-radius upper bound. We solve the reduced finite-ray dual problem for $m=1,\ldots,10$ and denote its solver objective by $\widehat D_{W_{48}}^{(m)}$. This solver objective numerically approximates the common primal--dual optimal value $U_{W_{48}}^{(m)}$ of the finite-ray SDP. All displayed runs terminate successfully, with the smallest eigenvalue
of the sampled positive-semidefinite constraint matrices at the level
of the solver tolerance.

Figure~\ref{fig:tight-copy-hierarchy} shows the normalized values $m\widehat D_{W_{48}}^{(m)}$.  Independent Schur--Weyl NH values are shown at the copy numbers for which the existing NH benchmark was evaluated on the same model, and the Holevo values are included only as asymptotic reference levels.  The finite-ray dual values decrease over the displayed copy numbers.  At $m=10$ we obtain
\[
 \begin{array}{c|cc}
  & r=0 & r=0.5\\
 \hline
 m\widehat D_{W_{48}}^{(m)} & 2.435688 & 1.957005\\
 mC_{\mathrm{NH}}^{(m)} & 2.429240 & 1.948554
 \end{array}
\]
The corresponding normalized differences are approximately $6.45\times10^{-3}$ and $8.45\times10^{-3}$. The computed finite-ray dual objectives lie close to the independently computed NH benchmark over the copy numbers considered.

\begin{figure}[H]
\centering
\begin{minipage}[t]{0.485\textwidth}
\centering
\includegraphics[width=\linewidth]{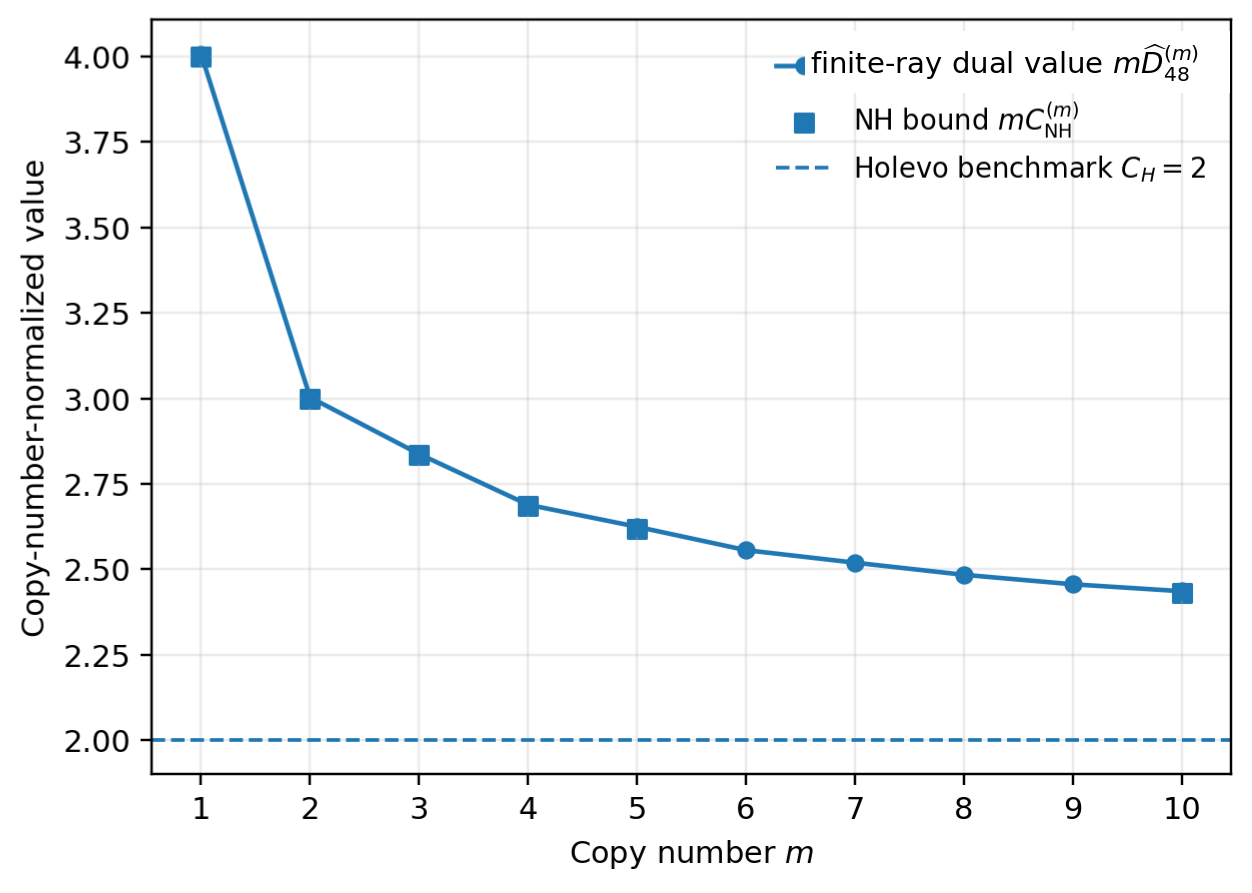}\\[-0.3em]
\small (a) $r=0$
\end{minipage}\hfill
\begin{minipage}[t]{0.485\textwidth}
\centering
\includegraphics[width=\linewidth]{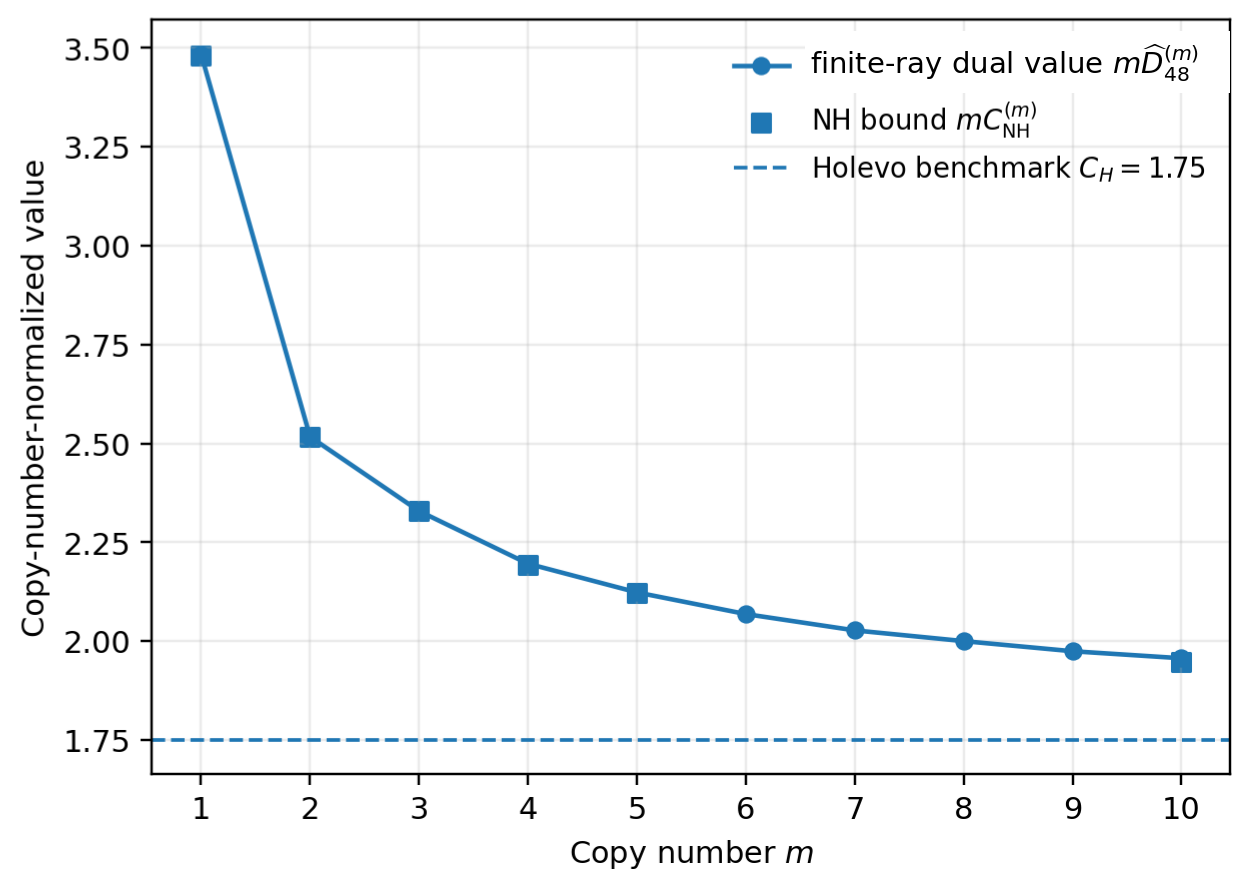}\\[-0.3em]
\small (b) $r=0.5$
\end{minipage}
\caption{Comparison after multiplication by the copy number $m$ on the fixed $L=48$ reference mesh, shown separately for the two radii. In each panel, the connected curve gives the numerical finite-ray dual values $m\widehat D_{W_{48}}^{(m)}$ for $m=1,\ldots,10$, square markers give independently computed NH bounds at $m=1,\ldots,5,10$, and the horizontal dashed line is the corresponding Holevo benchmark. The NH values are shown only at the copy numbers where they were independently evaluated.}
\label{fig:tight-copy-hierarchy}
\end{figure}

For $m=2$ on the $L=48$ mesh, Appendix~\ref{app:certified-finite-ray} gives certified enclosures of the finite-ray SDP optimum for $r=0$ and $r=0.5$, each with width below $5.1\times10^{-7}$.

\clearpage
\section{Applications to noisy multiparameter quantum sensing}
\label{sec:applications}

Whereas Sec.~\ref{sec:numerics} examines numerical accuracy and scaling on a reference qubit model, here we apply the reduced programs to three noisy sensing scenarios.
Model A estimates two rotation angles under known depolarizing noise; Model B
studies the same sensing task with known amplitude damping; and Model C jointly
estimates a phase and an unknown damping probability. In all three scenarios,
the same unknown parameters are shared by $m$ independently prepared probes,
noise acts independently on each copy, and arbitrary collective measurements
on the $m$-copy system are permitted. We use $G=I_2$ in the stated physical
parameter coordinates and multiply each finite-copy cost by $m$ before reporting it.

The exact optimum $C_{\mathrm{NH}}^{(m)}[I_2]$ is a mathematical lower
bound on the attainable finite-copy cost, and we report numerical approximations
to this optimum. Candidate POVM costs are obtained by reconstructing reduced
measurements from approximate finite-ray primal solutions
(Appendix~\ref{app:finite-ray}) and lifting them to the full $m$-copy system
with the adjoint channel $\Phi_m^*$. For each POVM, we optimize the locally
unbiased classical estimator using the inverse classical Fisher information
matrix and report the normalized trace of its MSE. We compare these costs
with the asymptotic Holevo reference $C_H$. The reconstructed strategies are
checked numerically for positivity, completeness, and local unbiasedness.

For reference, in the two-parameter qubit models below, the asymptotic
Holevo cost for $G=I_2$ can be evaluated by minimizing
\begin{equation}
 \sum_{i=1}^2\Tr(\rho X_i^2)
 +2\left|\operatorname{Im}\Tr(\rho X_1X_2)\right|
\label{eq:apps-holevo-evaluation}
\end{equation}
over Hermitian $X_i$ subject to
$\Tr(\rho X_i)=0$ and $\Tr(D_j X_i)=\delta_{ij}$.
The reference values below use the local states and derivatives displayed
for each model, rather than finite-copy solver outputs.

\subsection{Model A: two-axis rotation sensing with depolarizing noise}
\label{subsec:apps-v5-A}

\noindent\textit{Physical model.}
Consider a qubit initially in $|0\rangle$, on which two small, unknown
rotations are encoded about orthogonal transverse axes. Depolarizing noise,
with a known strength $\epsilon$, reduces the Bloch-vector length while
preserving rotational symmetry in the transverse plane. The two rotations
are associated with incompatible local observables, so their joint
precision depends on the allowed measurement strategy. This setting is
closely related to the few-copy collective-measurement experiment and
numerical study of Conlon et al.~\cite{Conlon2023}.

At the zero-rotation operating point, the state and tangent operators are
\begin{equation}
 \rho_{\mathrm A}=\frac{I+r\sigma_z}{2},\qquad
 D_x=-\frac r2\sigma_y,\qquad
 D_y=\frac r2\sigma_x,\qquad r=1-\epsilon\in(0,1).
 \label{eq:apps-v5-A-model}
\end{equation}
Here the unknown parameters are $(\theta_x,\theta_y)$, not $\epsilon$.
The full state family may be written
$\rho_{\boldsymbol\theta}=U_{\boldsymbol\theta}\rho_{\mathrm A}
U_{\boldsymbol\theta}^{\dagger}$ with
$U_{\boldsymbol\theta}=\exp[-i(\theta_x\sigma_x+\theta_y\sigma_y)/2]$.
For the numerical comparison below, different $\epsilon$ values represent
different, separately calibrated noise conditions.
For the Bloch vector $(0,0,r)$, the derivative constraints in
Eq.~\eqref{eq:apps-holevo-evaluation} fix the transverse components
of the two estimators. Their longitudinal components only increase
the cost. Evaluating the two transverse components gives
\begin{equation}
 C_{H,\mathrm A}=\frac{2}{r^2}+\frac{2}{r}
 =\frac{2(1+r)}{r^2},
\label{eq:apps-holevo-A}
\end{equation}
which equals $60$ and $220$ at $\epsilon=0.8$ and $0.9$,
respectively.

\noindent\textit{Numerical results.}
The general Schur--Weyl (SW) reduced NH SDP reproduces the known single-
and two-copy reference results. At $\epsilon=0.5$, its normalized three-copy value is
$3C_{\mathrm{NH}}^{(3)}[I_2]=12.716229$, consistent with the value reported in
Ref.~\cite{Conlon2023}; the four- to seven-copy results likewise agree
within the graphical precision of that work. To test the computational
advantage of the general reduction, we compared full-space and
SW-reduced versions of the same NH SDP under matched solver settings.
At $m=5$, three independent cold-process SCS runs gave median total-runtime
speedups of $7.58$ for $\epsilon=0.5$ and $4.03$ for $\epsilon=0.8$.
For $\epsilon=0.5$ and $\epsilon=0.8$, the solver-only speedups were
$37.14$ and $16.86$, respectively.
Figure~\ref{fig:apps-v5-A-speed} reports these two distinct timing measures.

\begin{figure}[H]
 \centering
 \includegraphics[width=0.72\linewidth]{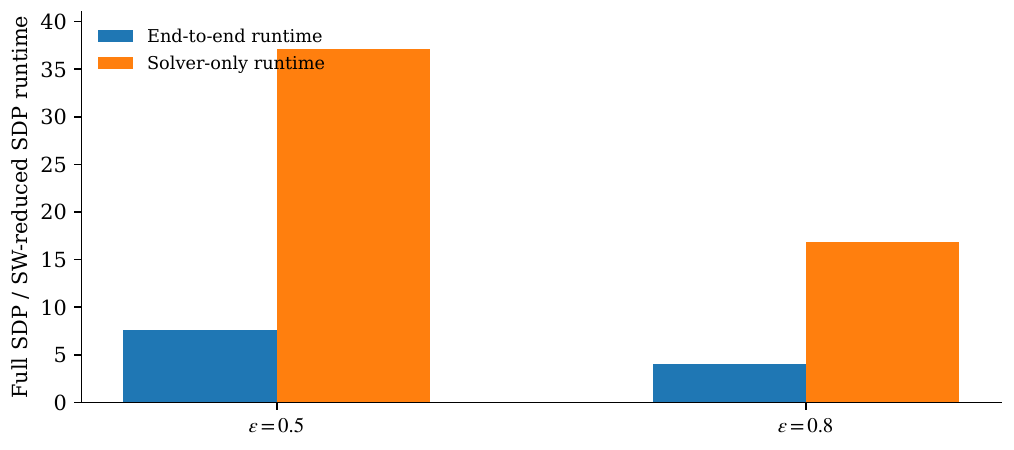}
 \caption{Model A: full-space versus general SW-reduced NH SDP at
 $m=5$, using SCS with three independent cold-process repetitions.
 End-to-end runtime and solver-only runtime are shown separately.
 No additional model-specific reduction is used in this comparison.}
 \label{fig:apps-v5-A-speed}
\end{figure}

For the large-copy calculation, we use additional symmetry specific to
this rotation model, beyond the general SW reduction. This calculation
is distinct from the matched full-space versus general SW-reduced SDP
benchmark at $m=5$; the additional simplification is not needed for
that benchmark. Figure
\ref{fig:apps-v5-A-large} and Table~\ref{tab:apps-v5-A-tail} show the
large-copy NH excess over the asymptotic Holevo reference,
$100(mC_{\mathrm{NH}}^{(m)}[I_2]-C_H)/C_H$, on a logarithmic scale. For instance, at $m=100$,
$100C_{\mathrm{NH}}^{(100)}[I_2]=60.041791$ versus $C_H=60$ for $\epsilon=0.8$, while
$100C_{\mathrm{NH}}^{(100)}[I_2]=222.377626$ versus $C_H=220$ for $\epsilon=0.9$.
Table~\ref{tab:apps-v5-A-tail} also reports the results for $\epsilon=0.5$.
The $m=100$ value is included to indicate the scale of the gap and is not
used to infer a convergence rate.

\begin{figure}[H]
 \centering
 \includegraphics[width=0.82\linewidth]{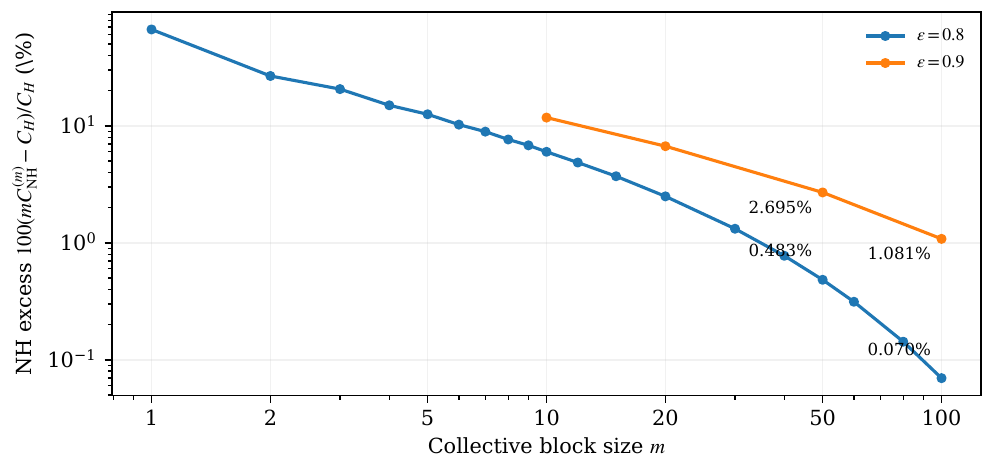}
 \caption{Model A: computed percentage excess of the finite-copy NH
 optimum estimate above the asymptotic Holevo reference. Logarithmic axes and
 high-copy numerical annotations make small differences visible.
 The largest-copy results exploit additional symmetry specific to this
 model; no convergence exponent is fitted.}
 \label{fig:apps-v5-A-large}
\end{figure}

\begin{table}[H]
 \centering\small
 \caption{Model A: relative percentage excess $100(mC_{\mathrm{NH}}^{(m)}[I_2]-C_H)/C_H$ over the
 asymptotic Holevo reference, not absolute mean-square error.}
 \label{tab:apps-v5-A-tail}
 \begin{tabular}{lrrr}
 \toprule
 Known noise strength & $m=20$ & $m=50$ & $m=100$\\
 \midrule
 $\epsilon=0.5$ & $0.033906\%$ & $9.1285\times10^{-5}\%$ & $\sim 3\times10^{-8}\%$\\
 $\epsilon=0.8$ & $2.494\%$ & $0.483\%$ & $0.070\%$\\
 $\epsilon=0.9$ & $6.694\%$ & $2.695\%$ & $1.081\%$\\
 \bottomrule
 \end{tabular}
\end{table}

These data extend a physically motivated finite-copy calculation far
beyond the original few-copy regime. The matched full/SW tests quantify
the advantage of the general reduction, while the larger-copy study
illustrates what can be achieved when additional model symmetry is
available.

\subsection{Model B: two-axis rotation sensing with amplitude damping}
\label{subsec:apps-v5-B}

\noindent\textit{Physical model.}
Unlike isotropic depolarization, amplitude damping describes directional
relaxation, for example decay of a qubit from $|1\rangle$ to $|0\rangle$.
Related multiparameter rotation-estimation problems under amplitude-damping
noise have been studied in Ref.~\cite{ConlonLamAssad2023}.
To study the two-rotation problem under this different physical noise,
we prepare $|1\rangle$, encode the same two unknown transverse rotations,
and subsequently apply amplitude damping with a known probability
$p$. The channel has Kraus operators
\begin{equation}
 K_0=|0\rangle\langle0|+\sqrt{1-p}|1\rangle\langle1|,\qquad
 K_1=\sqrt p\,|0\rangle\langle1|.
 \label{eq:apps-v5-B-channel}
\end{equation}
The corresponding zero-rotation local model is
\begin{equation}
 \rho_{\mathrm B}=\frac{I+(2p-1)\sigma_z}{2},\qquad
 D_x=\frac{\sqrt{1-p}}2\sigma_y,\qquad
 D_y=-\frac{\sqrt{1-p}}2\sigma_x.
 \label{eq:apps-v5-B-model}
\end{equation}
Unlike Model C below, $p$ is calibrated and not estimated. We emphasize
$p=1/2$: the operating-point state is maximally mixed, but its
rotation derivatives remain nonzero, so both angles are locally
identifiable. At $p=1/2$, the state is $I/2$
and the tangent Bloch vectors have length $1/\sqrt2$ in orthogonal
directions. Equation~\eqref{eq:apps-holevo-evaluation} then gives
$C_{H,\mathrm B}=2+2=4$; its imaginary term vanishes because
the Bloch vector at this operating point is zero.

\noindent\textit{Numerical results.}
We independently constructed the amplitude-damping SW blocks and solved
the general SW-reduced NH SDP at $p=1/2$ through $m=20$. The $m=50$
and $m=100$ computations additionally use symmetry specific to this
operating point; they are distinct from the general SW-reduced
calculations at $m\le20$. Figure
\ref{fig:apps-v5-B-large} and Table~\ref{tab:apps-v5-B-large} show the
normalized NH cost, with the asymptotic Holevo reference $C_H=4$.
At $m=100$, the computed normalized NH value is $4.284746$, about $7.12\%$ above that
reference. Since the exact NH optimum is a mathematical lower bound on the
attainable finite-copy value, this numerical result indicates a remaining
precision limitation even for a hypothetical 100-copy collective measurement.

\begin{table}[H]
 \centering\small
 \caption{Model B at $p=1/2$: computed normalized finite-copy NH values approximating $mC_{\mathrm{NH}}^{(m)}[I_2]$.
 The asymptotic Holevo reference is $C_H=4$ for all rows.
 The largest-copy entries use additional model-specific symmetry.}
 \label{tab:apps-v5-B-large}
 \begin{tabular}{rrrrrrrr}
 \toprule
 $m$ & 1 & 3 & 5 & 10 & 20 & 50 & 100\\
 \midrule
 $mC_{\mathrm{NH}}^{(m)}[I_2]$ & 8.000000 & 5.671009 & 5.238619 & 4.858170 & 4.612472 & 4.395643 & 4.284746\\
 \bottomrule
 \end{tabular}
\end{table}

\begin{figure}[H]
 \centering
 \includegraphics[width=0.88\linewidth]{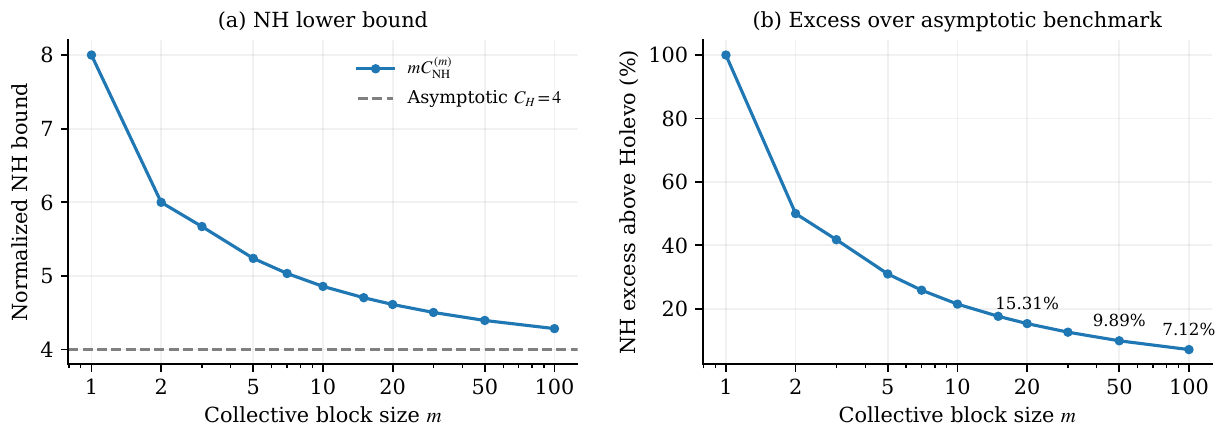}
 \caption{Model B, $p=1/2$: (a) normalized NH cost and horizontal
 asymptotic Holevo reference $C_H=4$; (b) percentage NH excess above
 that reference, with the $m=20,50,100$ values labelled.}
 \label{fig:apps-v5-B-large}
\end{figure}

To examine candidate finite-copy measurement performance, we also
solved finite-direction SW-reduced tight programs for $m=3,5,7$,
recovered candidate collective POVMs, and optimized their classical
locally unbiased estimators. At $m=3,5,7$, respectively, the numerically reconstructed candidate POVMs
have normalized computed costs $5.671030$, $5.238640$, and $5.036238$,
compared with computed NH optimum estimates $5.671009$, $5.238619$, and
$5.033869$ (Fig.~\ref{fig:apps-v5-B-povm}).

\begin{figure}[H]
 \centering
 \includegraphics[width=0.72\linewidth]{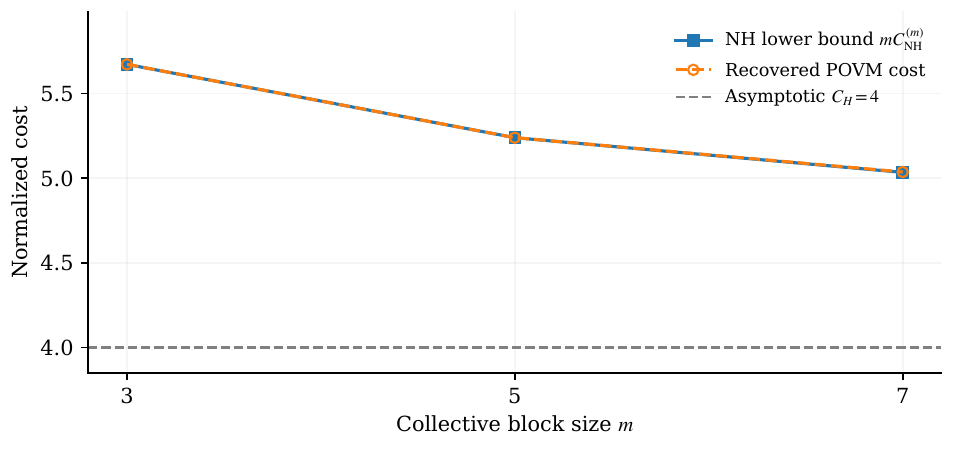}
 \caption{Model B, $p=1/2$: computed normalized NH optimum estimates versus
 numerically reconstructed candidate collective-POVM costs at $m=3,5,7$.
 The horizontal reference $C_H=4$ is asymptotic, not a finite-copy
 measurement result.}
 \label{fig:apps-v5-B-povm}
\end{figure}

Table~\ref{tab:apps-v5-B-diagnostics} reports the solver statuses and numerical
residuals for completeness and local unbiasedness. For $m=3$ and $m=7$, the
solver returned \texttt{optimal\_inaccurate}, indicating an accuracy warning.
\begin{table}[H]
 \centering\small
 \caption{Model B numerical solver diagnostics. The repair term is not a certified interval width.}
 \label{tab:apps-v5-B-diagnostics}
 \begin{tabular}{rrlrrr}
 \toprule
 $m$ & $L$ & Solver status & Completeness error & Unbiasedness error & Repair term\\
 \midrule
 3 & 96 & \texttt{optimal\_inaccurate} & $4.219\cdot10^{-15}$ & $6.550\cdot10^{-15}$ & $0.443365$\\
 5 & 96 & \texttt{optimal} & $7.106\cdot10^{-15}$ & $9.548\cdot10^{-15}$ & $0.423962$\\
 7 & 48 & \texttt{optimal\_inaccurate} & $3.331\cdot10^{-15}$ & $4.441\cdot10^{-16}$ & $2.311394$\\
 \bottomrule
 \end{tabular}
\end{table}

The candidate POVM costs are close to the computed NH values at $m=3,5,7$.
The finite-mesh correction is too conservative to yield narrow rigorous
intervals for the tight value in these calculations.

\subsection{Model C: joint phase and unknown damping estimation}
\label{subsec:apps-v5-C}

\noindent\textit{Physical model.}
In Models A and B the noise parameters are assumed calibrated. In
practice, a sensor may have to estimate its signal phase and noise
strength simultaneously. The simultaneous estimation of phase and
amplitude-damping parameters has also been studied in qubit
models~\cite{Razavian2020}. Model C addresses this situation by preparing
$|+\rangle=(|0\rangle+|1\rangle)/\sqrt2$, encoding a phase
$\phi$ with $U_\phi=\exp(-i\phi\sigma_z/2)$, and then applying the
amplitude-damping channel of Eq.~\eqref{eq:apps-v5-B-channel} with
unknown probability $p$. The output family is
\begin{equation}
 \rho_{\phi,p}=\frac12\left[I+
 \sqrt{1-p}\big(\cos\phi\,\sigma_x+\sin\phi\,\sigma_y\big)
 +p\sigma_z\right].
 \label{eq:apps-v5-C-model}
\end{equation}
At the operating point $\phi=0$, the local derivatives are
\begin{equation}
 D_\phi=\frac{\sqrt{1-p}}2\sigma_y,\qquad
 D_p=-\frac{1}{4\sqrt{1-p}}\sigma_x+\frac12\sigma_z.
 \label{eq:apps-v5-C-derivatives}
\end{equation}
Thus phase variation changes the orientation of the Bloch vector,
whereas varying the unknown damping parameter also changes the state
spectrum. Accordingly, the probability of obtaining a particular SW
spin sector depends on $p$ but not on $\phi$. This sector label is
itself a source of information about the unknown noise strength.
To calculate the Holevo reference, set $q=1-p$ and use
$X_i=\boldsymbol{x}_i\cdot\boldsymbol{\sigma}
-(\boldsymbol{s}\cdot\boldsymbol{x}_i)I$ with
$\boldsymbol{s}=(\sqrt q,0,p)$. The derivative constraints allow
$\boldsymbol{x}_\phi=(a,1/\sqrt q,a/(2\sqrt q))$ and
$\boldsymbol{x}_p=(b,0,1+b/(2\sqrt q))$.
Substituting these expressions into the Holevo objective
\eqref{eq:apps-holevo-evaluation} yields a quadratic contribution
$(1+p)a^2/4$; the absolute-value contribution does not depend on $a$.
Thus minimization first sets $a=0$, leaving the one-variable problem
\begin{equation}
 \min_{b\in\mathbb R}\left\{\frac1q+1-p^2
 +\frac{1+p}{4}b^2+q^{3/2}b
 +2\left|1+\frac{1-2p}{2\sqrt q}b\right|\right\}.
\label{eq:apps-holevo-C-min}
\end{equation}
On the branch where the expression inside the absolute value is
positive, differentiating the objective with respect to $b$ gives the
stationarity condition
\[
 \frac{1+p}{2}b+q^{3/2}+\frac{1-2p}{\sqrt q}=0.
\]
Its solution is
$b=2(-p^2+4p-2)/((1+p)\sqrt{1-p})$. At $p=0.3$ and $p=0.7$,
the expression inside the absolute value at this solution equals
$277/455$ and $193/255$, respectively. Both are positive, so the
stationarity equation was evaluated on the correct branch. Since the
objective in Eq.~\eqref{eq:apps-holevo-C-min} is convex in $b$, this
stationary point is a global minimizer at each of these two operating
points. Consequently,
\begin{equation}
 C_{H,\mathrm C}(p)=\frac{p(17-24p+8p^2)}{1-p^2}
 \quad(p=0.3,\ 0.7),
\label{eq:apps-holevo-C}
\end{equation}
which yields $1578/455$ and $1442/255$, respectively.

\noindent\textit{Numerical results.}
We studied two distinct operating points, $p=0.3$ and $p=0.7$, at
$m=2,3,5$, solving general SW-reduced NH problems and finite-direction
tight programs. The finite-direction tight programs supply candidate
collective POVMs whose
classical locally unbiased estimators are further optimized. Table
\ref{tab:apps-v5-C-values} and Fig.~\ref{fig:apps-v5-C-precision}
compare their normalized costs. The asymptotic normalized Holevo references are
$C_H(0.3)\simeq3.468132$ and $C_H(0.7)\simeq5.654902$, shown as
horizontal lines rather than assigned to any finite-copy row.

\begin{table}[H]
 \centering\small
 \caption{Model C: computed normalized NH optimum estimates approximating $mC_{\mathrm{NH}}^{(m)}[I_2]$ and
 the numerical costs of reconstructed candidate collective POVMs. The Holevo reference
 is asymptotic and is specified separately in the text.}
 \label{tab:apps-v5-C-values}
 \begin{tabular}{crrr}
 \toprule
 $p$ & $m$ & $mC_{\mathrm{NH}}^{(m)}[I_2]$ & POVM cost\\
 \midrule
 $0.3$ & 2 & 3.569273 & 3.569621\\
       & 3 & 3.514602 & 3.515776\\
       & 5 & 3.479626 & 3.480075\\
 \midrule
 $0.7$ & 2 & 5.801422 & 5.802452\\
       & 3 & 5.725225 & 5.725818\\
       & 5 & 5.672978 & 5.673249\\
 \bottomrule
 \end{tabular}
\end{table}

\begin{figure}[H]
 \centering
 \includegraphics[width=0.88\linewidth]{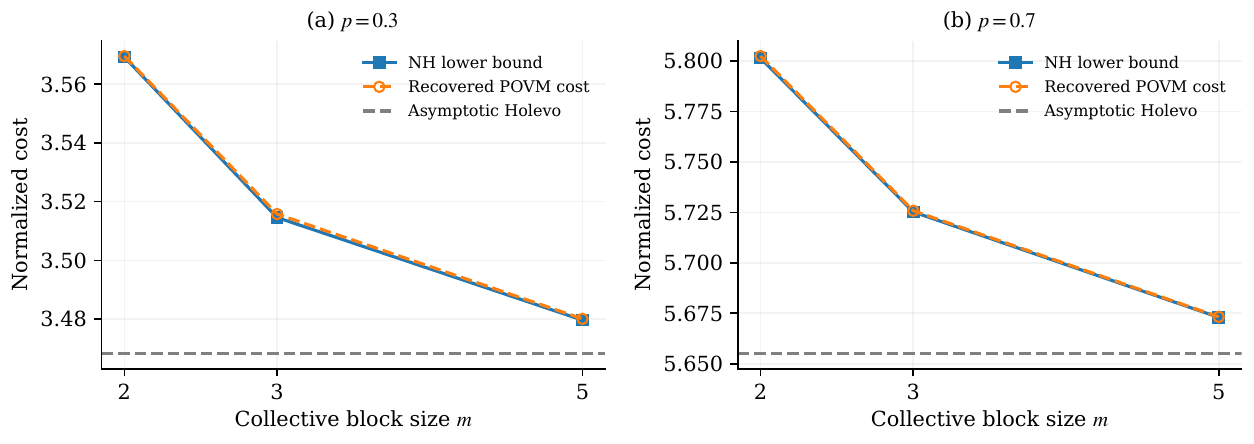}
 \caption{Model C: finite-copy NH benchmark and numerically
 reconstructed POVM performance for $p=0.3,0.7$. Horizontal dashed
 lines are the corresponding asymptotic Holevo references.}
 \label{fig:apps-v5-C-precision}
\end{figure}

To quantify the information carried by the sector label itself, write
$p_\lambda(p):=\Tr\bar R_{\lambda,p}$ for its probability in the reduced
$m$-copy state. The classical Fisher information and its fraction of the
total $m$-copy SLD $pp$ information are
\begin{equation}
 J_{\mathrm{sector},pp}^{(m)}(p)
 :=\sum_{\lambda:p_\lambda(p)>0}
 \frac{[\partial_p p_\lambda(p)]^2}{p_\lambda(p)},
 \qquad
 f_m(p):=\frac{J_{\mathrm{sector},pp}^{(m)}(p)}
 {m J_{\mathrm{SLD},pp}^{(1)}(p)}.
\label{eq:apps-v5-C-sector-FI}
\end{equation}
At $\phi=0$, the Bloch vector of Eq.~\eqref{eq:apps-v5-C-model} is
$(\sqrt{1-p},0,p)$; its single-copy SLD information is
$J_{\mathrm{SLD},pp}^{(1)}(p)=(1+p)/[4p(1-p)]$ for $0<p<1$.
For example, at $m=2$ the antisymmetric sector has probability
$p_-(p)=p(1-p)/4$, while the symmetric sector has probability $1-p_-(p)$.
Therefore,
\[
 J_{\mathrm{sector},pp}^{(2)}(p)
 =\frac{[\partial_p p_-(p)]^2}{p_-(p)[1-p_-(p)]},
 \qquad \partial_p p_-(p)=\frac{1-2p}{4}.
\]
At $p=0.3$ these expressions give
$J_{\mathrm{sector},pp}^{(2)}=0.201030\ldots$ and
$f_2(0.3)=0.064948\ldots$, or $6.49\%$.
For $m=3,5$ the corresponding fractions at $p=0.3$ are
$9.17\%$ and $11.65\%$; at $p=0.7$ the fractions for
$m=2,3,5$ are $4.97\%$, $7.01\%$, and $8.91\%$.
Figure~\ref{fig:apps-v5-C-sector} displays these results. These
percentages quantify a Fisher-information contribution, not
a proportional decrease in the final MSE.

\begin{figure}[H]
 \centering
 \includegraphics[width=0.72\linewidth]{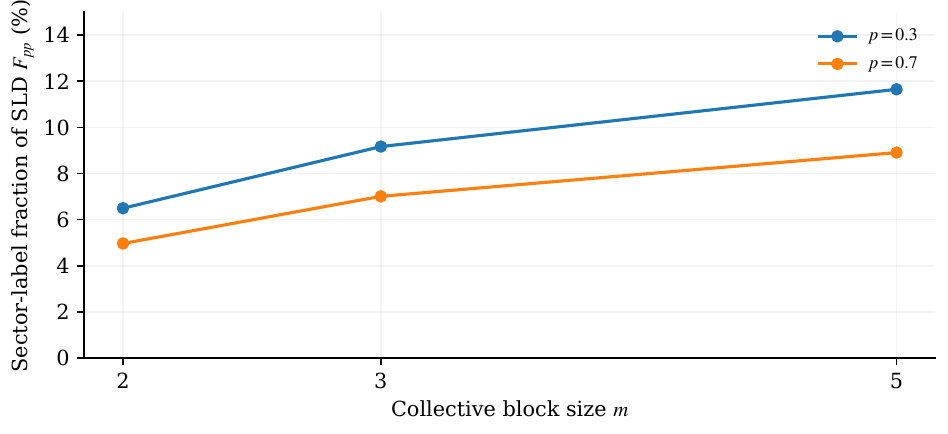}
 \caption{Model C: the classical Fisher information about the unknown
 damping probability obtained from the Schur-sector label alone,
 divided by the total $m$-copy SLD $pp$ information.}
 \label{fig:apps-v5-C-sector}
\end{figure}

Although $p=0.3$ and $p=0.7$ give the same one-copy state purity,
$\|\boldsymbol r\|^2=1-p+p^2=0.79$, their parameter sensitivities
and achievable precision costs differ. Model C therefore illustrates
why both the SW sector weights and their parameter derivatives
must be included when the noise strength itself is unknown.

\section{Discussion}
\label{sec:discussion}

The Schur--Weyl reduced model allows us to study exact finite-copy attainable precision without explicitly representing the permutation multiplicity spaces. For every POVM on the full $m$-copy model, there is a POVM on the reduced model, and conversely, with the same outcome probabilities for every $\theta$. Therefore, Schur--Weyl reduction preserves the optimization over arbitrary collective POVMs without approximating the tight precision value. The reduced tight formulation still involves a continuum of dual constraints. For any sufficiently fine finite-ray mesh, the constructive feasibility and uniform dual bounds establish quadratic convergence of its tight-value approximation at each fixed copy number. Within the same reduced model, the Nagaoka--Hayashi SDP provides a simpler relaxation that allows computations at larger copy numbers.

More generally, since the reduction preserves the statistical experiment itself, the reduced model can also be used to study other measurement optimization problems on i.i.d. quantum state families. The common reduction argument may extend to further conic formulations when their cone structures and constraints are compatible with the reduction and recovery maps.

The resulting complexity statements concern different asymptotic directions. For fixed one-copy dimension, Corollary~\ref{cor:copy-number-scaling} gives polynomial growth of the reduced operator representation with the copy number. For a fixed finite-copy tight problem, Corollary~\ref{cor:reference-quadratic} establishes the sufficient ray scaling $O(\epsilon^{-d/2})$ for theoretical enclosures on sufficiently fine reference meshes. The constants in the second statement have not been shown to be uniform in $m$, so these two scalings do not give a joint complexity theorem.

Our results apply to finite-dimensional i.i.d. quantum state models with arbitrary collective POVMs. The fixed-mesh calculations through ten copies provide an initial numerical view of the finite-ray dual values relative to the NH and Holevo benchmarks; they do not determine an asymptotic convergence law. Two extensions are therefore natural within the present setting: sector-adapted meshes that exploit the block dependence of the continuous minimizers, and a systematic study of the approach of the finite-copy attainable precision value to the asymptotic Holevo regime.

\appendix

\section{Schur--Weyl representation details and recovery channels}
\label{app:sw-details}

This appendix gives supporting details for the finite-copy statistical model and its Schur--Weyl reduction. We first verify that the standing linear-independence assumption is inherited by every copy number.
\begin{lemma}[Independence of tensor-power tangents]
\label{app:tangent-independence}
If $D_1,\ldots,D_d$ are linearly independent, then $D_1^{(m)},\ldots,D_d^{(m)}$ are linearly independent for every $m\ge1$.
\end{lemma}
\begin{proof}
Tracing out copies $2,\ldots,m$ in Eq.~\eqref{eq:m-copy-tangents} gives $\Tr_{2,\ldots,m}D_i^{(m)}=D_i$ because $\Tr\rho=1$ and $\Tr D_i=0$. Hence $\sum_ic_iD_i^{(m)}=0$ implies $\sum_ic_iD_i=0$, and therefore $c_i=0$ for all $i$.
\end{proof}

For completeness, recall the representation facts underlying the reduction. A partition $\lambda\vdash m$ is a non-increasing sequence of positive integers summing to $m$; $\ell(\lambda)$ denotes its number of rows. The restriction $\ell(\lambda)\le n$ in Eq.~\eqref{eq:young-index-set} is the usual one from Schur--Weyl duality on $(\mathbb C^n)^{\otimes m}$. The multiplicity $f_\lambda=\dim\cV_\lambda$ is given by the hook-length formula~\cite{Sagan2001}; here $\lambda'$ denotes the conjugate partition,
\[
 f_\lambda
 =\frac{m!}{\prod_{(r,c)\in\lambda}h_\lambda(r,c)},
 \qquad
 h_\lambda(r,c)=\lambda_r-c+\lambda_c'-r+1.
\]
The two commuting group actions take the form below, where $P_\pi$ is the unitary that permutes the $m$ tensor factors according to $\pi\in S_m$:
\[
 U^{\otimes m}\simeq U_\lambda(U)\otimes I_{\cV_\lambda},
 \qquad
 P_\pi\simeq I_{\cU_\lambda}\otimes V_\lambda(\pi)
\]
on each sector. Thus operators commuting with all permutations act trivially on $\cV_\lambda$, which gives Eqs.~\eqref{eq:block-state}--\eqref{eq:block-tangent}.

We next verify the channel identities used in Sec.~\ref{sec:sw-reduction}. The map $\Phi_m$ is block compression followed by partial trace over $\cV_\lambda$, summed over orthogonal sectors, and is therefore CPTP. The map $\Psi_m$ dephases across $\bigoplus_\lambda\cU_\lambda$ and attaches $\tau_\lambda=I_{\cV_\lambda}/f_\lambda$ in each sector, so it is also CPTP. Using the state block decomposition in Eq.~\eqref{eq:block-state} and the multiplicity weights in Eq.~\eqref{eq:weighted-blocks},
\begin{align*}
 \Phi_m(\rho_\theta^{\otimes m})
 &=\bigoplus_\lambda f_\lambda R_{\lambda,\theta},\\
 \Psi_m\Phi_m(\rho_\theta^{\otimes m})
 &=\bigoplus_\lambda f_\lambda R_{\lambda,\theta}\otimes\frac{I_{\cV_\lambda}}{f_\lambda}
 =\rho_\theta^{\otimes m}.
\end{align*}
Differentiation gives the tangent identities in Eq.~\eqref{eq:tangent-recovery}. These identities explain the factors $f_\lambda$ in the weighted blocks $\bar R_{\lambda,\theta}$ and $\bar D_{i,\lambda}$.

\section{Preservation of the tight dual value under Schur--Weyl reduction}
\label{app:tight-cone}

The common primal reduction was proved in Theorem~\ref{thm:common-reduction}. Here we record the corresponding dual argument needed for Corollary~\ref{cor:tight-reduction}. On the full model, the Hayashi--Ouyang tight dual~\cite{HayashiOuyang2023} has objective $\Tr a+\Tr S$ and constraint
\begin{equation}
 \Gamma_m(a,S)
 :=\widetilde G\otimes\rho^{\otimes m}
 -\sum_{i,j=1}^d a_{ij}Q_i\otimes D_j^{(m)}
 -|0\rangle\langle0|\otimes S
 \in\cS_{\SEP}^*(\cR:\cH^{\otimes m}).
\label{eq:Gamma-full-app}
\end{equation}
Under the regularity assumptions of Sec.~\ref{subsec:local-model}, its optimal value equals $C^{(m)}[G]$~\cite{HayashiOuyang2023}.

Let $(a,S)$ satisfy the full dual constraint in Eq.~\eqref{eq:Gamma-full-app} and set $S_{\mathrm{red}}:=\Phi_m(S)=\bigoplus_\lambda\bar S_\lambda$. Using the exact state and tangent recovery identities in Eqs.~\eqref{eq:family-recovery}--\eqref{eq:tangent-recovery},
\begin{equation}
 (\id_{\cR}\otimes\Phi_m)\Gamma_m(a,S)
 =\bigoplus_\lambda\bar\Gamma_\lambda(a,\bar S_\lambda).
\label{eq:dual-block-transfer}
\end{equation}
A completely positive map on the physical factor sends the dual separable cone into the corresponding reduced dual cone. Indeed, if $Z\in\cS_{\SEP}^*(\cR:\cH^{\otimes m})$, $A\in\Sym_+(\cR)$, and $B\succeq0$, then
\[
 \Tr[(\id_{\cR}\otimes\Phi_m)(Z)(A\otimes B)]
 =\Tr[Z(A\otimes\Phi_m^*(B))]\ge0.
\]
Thus Eq.~\eqref{eq:dual-block-transfer} takes a full feasible dual pair to a reduced feasible pair. Since $\Phi_m$ is trace preserving,
\[
 \Tr S_{\mathrm{red}}=\Tr\Phi_m(S)=\Tr S,
\]
so this transfer also preserves the dual objective $\Tr a+\Tr S$.

Conversely, suppose that $(a,\{\bar S_\lambda\})$ satisfies every reduced dual constraint. Put
\[
 Z_{\mathrm{red}}:=\bigoplus_\lambda
    \bar\Gamma_\lambda(a,\bar S_\lambda),
 \qquad
 S:=\Psi_m\!\left(\bigoplus_\lambda\bar S_\lambda\right).
\]
The block conditions in Eq.~\eqref{eq:red-dual-constraint} imply $Z_{\mathrm{red}}\in\cS_{\SEP}^*(\cR:\cK_m)$. The recovery identities for the state and tangents then give the operator equality
\begin{equation}
 (\id_{\cR}\otimes\Psi_m)(Z_{\mathrm{red}})
 =\widetilde G\otimes\rho^{\otimes m}
  -\sum_{i,j}a_{ij}Q_i\otimes D_j^{(m)}
  -|0\rangle\langle0|\otimes S
 =\Gamma_m(a,S).
\label{eq:dual-reverse-transfer}
\end{equation}
Because $\Psi_m$ is completely positive, applying the same dual-cone argument used above, with $\Psi_m^*$ in place of $\Phi_m^*$, shows that $\Gamma_m(a,S)\in\cS_{\SEP}^*(\cR:\cH^{\otimes m})$. This is the required full dual feasibility. Finally, trace preservation gives
\[
 \Tr S=\Tr\Psi_m\!\left(\bigoplus_\lambda\bar S_\lambda\right)
 =\sum_\lambda\Tr\bar S_\lambda.
\]
The multiplier $a$ is unchanged, so the reverse transfer also preserves the objective. Both transfers preserve objective values and feasibility, and hence the full and reduced dual optimal values are equal. The full dual equals $C^{(m)}[G]$ by the cited full-space strong-duality result; together with Theorem~\ref{thm:common-reduction}, this establishes Eq.~\eqref{eq:red-equality}.

\section{Finite-ray reconstruction and strong duality}
\label{app:finite-ray}

Let $S_+:=\{s:w_{s,0}\neq0\}$. If $w_{s,0}=0$, then $\langle w_s|Q_i|w_s\rangle=0$ for every $i$, while $\langle w_s|\widetilde G|w_s\rangle>0$ for every nonzero direction in the parameter subspace. For $w_{s,0}=0$, the matrices $X_{s,\lambda}$ contribute to neither the completeness nor the derivative constraints. Any nonzero $X_{s,\lambda}\succeq0$ increases the objective. Therefore, an optimal solution may be chosen with $X_{s,\lambda}=0$ for these rays.

For $s\in S_+$ define
\[
 \Pi_{s,\lambda}:=|w_{s,0}|^2X_{s,\lambda},
 \qquad
 \widehat\theta_i(s):=\frac{w_{s,i}}{w_{s,0}}.
\]
Define the POVM operator on the reduced space $\cK_m$ by
\[
 \Pi_s:=\bigoplus_{\lambda\in\Ymn}\Pi_{s,\lambda}.
\]
All sums over $s$ in this reconstruction are over $S_+$. The completeness constraint~\eqref{eq:finite-completeness} gives $\sum_s\Pi_{s,\lambda}=I_{\cU_\lambda}$ for every $\lambda$. Thus, $\{\Pi_s\}_{s\in S_+}$ is a POVM on $\cK_m$. The derivative constraint~\eqref{eq:finite-LU} gives
\[
 \sum_s\widehat\theta_i(s)\Tr[D_{j,\mathrm{red}}^{(m)}\Pi_s]=\delta_{ij}.
\]
The objective in Eq.~\eqref{eq:finite-ray-primal} is the weighted second moment of these estimator values. To impose the zero mean condition for a locally unbiased estimator, define
\[
 \mu_i:=\sum_s\widehat\theta_i(s)\Tr[\rho_{\mathrm{red}}^{(m)}\Pi_s],
 \qquad
 \widehat\theta_i'(s):=\widehat\theta_i(s)-\mu_i.
\]
Because $\sum_s\Tr[D_{j,\mathrm{red}}^{(m)}\Pi_s]=\Tr D_{j,\mathrm{red}}^{(m)}=0$, this recentering leaves the derivative normalization unchanged and enforces zero mean. Moreover,
\[
 V^{(m)}[\Pi,\widehat\theta']
 =V^{(m)}[\Pi,\widehat\theta]-\mu\mu^T,
\]
so the weighted MSE cannot increase. Proposition~\ref{prop:stat-equivalence} then transfers the reconstructed strategy to the full $m$-copy model. This establishes the finite-ray upper bound $C^{(m)}[I_d]\le U_W^{(m)}$ in Eq.~\eqref{eq:upper-bound}. Under the standing regularity assumptions, one can choose a finite set of rays for which the finite-ray primal constraints admit a solution. These rays can be constructed from a finite-outcome locally unbiased strategy using vectors proportional to $(1,\widehat\theta(x))$.

The finite-ray primal and dual are ordinary finite-dimensional SDPs. If the finite-ray primal constraints admit a solution, dual Slater is immediate: choose $a=0$ and $\bar S_\lambda=-tI_{\cU_\lambda}$ with $t>0$. Substituting this choice into the sampled constraints~\eqref{eq:finite-ray-dual-rays} gives
\[
 \langle w_s|\bar\Gamma_\lambda(0,-tI)|w_s\rangle
 =\langle w_s|\widetilde G|w_s\rangle\bar R_\lambda
 +t|w_{s,0}|^2I_{\cU_\lambda}\succ0.
\]
Strict positivity follows from the second term when $w_{s,0}\ne0$ and from $G\succ0$ together with $\bar R_\lambda\succ0$ when $w_{s,0}=0$. Standard SDP dual Slater therefore gives zero duality gap. Hence the dual objective can approach the finite-ray primal optimal value arbitrarily closely while satisfying Eq.~\eqref{eq:finite-ray-dual-rays}. This is the form used in Theorem~\ref{thm:global-mesh-bound}.

\section{Quadratic repair, projective-to-affine estimates, and general weights}
\label{app:quadratic-proof}

\subsection{Exact completion of squares and sectorwise repair for identity weight}
\label{app:identity-repair}

The scalar function $q_{\lambda,u}(y)$ and its minimizer $y_{\lambda,u}=\frac12 b_{\lambda,u}$ are defined in Sec.~\ref{subsec:quadratic-geometry}. The main text obtains the exact quadratic identity~\eqref{eq:quadratic-identity} by completing the square and uses it to establish the sectorwise repair~\eqref{eq:sector-repair-minimum}. For the global estimate, it remains to bound the locations of these minimizers uniformly over sectors and unit vectors $u$.

The weighted-block definitions in Eq.~\eqref{eq:weighted-blocks} give $\bar R_\lambda=f_\lambda R_\lambda$ and $\bar D_{i,\lambda}=f_\lambda D_{i,\lambda}$. The multiplicity factor $f_\lambda$ cancels when the tangent operators are normalized by $\bar R_\lambda^{-1/2}$. Together with the state and tangent block decompositions in Eqs.~\eqref{eq:block-state}--\eqref{eq:block-tangent}, this gives
\begin{equation}
 B_i^{\mathrm{full}}(a)
 \cong\bigoplus_\lambda B_{i,\lambda}\otimes I_{\cV_\lambda}.
\label{eq:B-full-block}
\end{equation}
Hence, for every unit $u\in\cU_\lambda$,
\begin{align}
 \|b_{\lambda,u}\|^2
 &=\sum_i\langle u|B_{i,\lambda}|u\rangle^2\\
 &\le\left\langle u\middle|\sum_iB_{i,\lambda}^2\middle|u\right\rangle\\
 &\le\left\|\sum_iB_{i,\lambda}^2\right\|
 \le4\bigl(C_2^{\mathrm{HO}}(a)\bigr)^2,
\end{align}
which gives Eq.~\eqref{eq:identity-radius}.

\subsection{Projective-to-affine bridge and reference mesh}
\label{app:projective-affine}

\begin{lemma}[Projective-to-affine bridge]
\label{lem:projective-affine}
Let $C\ge0$ and let $\delta:=\delta(W)$ be the projective radius defined in Eq.~\eqref{eq:projective-radius}. If
\[
 \delta<\frac{2}{\sqrt{1+C^2}},
\]
then every $y\in\mathbb R^d$ with $\|y\|\le C$ has a sampled affine point $z\in\mathcal Y_W$. Its Euclidean distance from $y$ is bounded by the projective-to-affine error function $\eta(C,\delta)$ defined in Eq.~\eqref{eq:eta}:
\[
 \|y-z\|\le\eta(C,\delta)
 =\frac{(1+C^2)\delta}{2\sqrt{1-\delta^2/4}-C\delta}.
\]
\end{lemma}

\begin{proof}
Write $x(y)=(1,y)/\sqrt{1+\|y\|^2}$. By the definition of the projective covering radius, choose normalized $w\in W$ with $\||x(y)\rangle\langle x(y)|-|w\rangle\langle w|\|_1\le\delta$ and choose the sign of $w$ so that $\langle x(y),w\rangle\ge0$. Set $\gamma:=\arcsin(\delta/2)$ and let $\beta:=\arccos\langle x(y),w\rangle\in[0,\gamma]$ be the actual angle. Then $w=\cos\beta\,x(y)+\sin\beta\,q$ for some unit $q\perp x(y)$ (with arbitrary such $q$ if $\beta=0$). Define $\alpha:=\arctan\|y\|$. The assumption on $\delta$ implies $\alpha+\gamma<\pi/2$. Therefore,
\[
 w_0\ge\cos(\alpha+\beta)\ge\cos(\alpha+\gamma)>0.
\]
Hence, the affine coordinate $z=w_{1:d}/w_0$ is well defined. The identity
\[
 \frac14\bigl\||x(y)\rangle\langle x(y)|-|x(z)\rangle\langle x(z)|\bigr\|_1^2
 =\frac{\|y-z\|^2+\|y\|^2\|z\|^2-(y^Tz)^2}{(1+\|y\|^2)(1+\|z\|^2)}
\]
and the nonnegativity of the Gram-determinant term give, since $w_0>0$ and thus $w=x(z)$,
\begin{align*}
 \|y-z\|
 &\le\frac{\sin\beta}{\cos\alpha\,w_0}
 &&\text{(nonnegative Gram-determinant term)},\\
 &\le\frac{\delta/2}{\cos\alpha\cos(\alpha+\gamma)}
 &&\text{(}\sin\beta\le\delta/2,\quad w_0\ge\cos(\alpha+\gamma)\text{)}.
\end{align*}
The right-hand side is increasing for $\alpha\in[0,\arctan C]$; evaluating at $\alpha=\arctan C$ gives the claimed formula.
\end{proof}

Following the recursive construction of Hayashi and Ouyang~\cite[Sec.~5.4]{HayashiOuyang2023}, we define the reference mesh for even integers $L\ge4$ by the indexed families of unit vectors
\[
 \begin{aligned}
 W_{L,1}&:=\bigl\{(\cos(2\pi j/L),\sin(2\pi j/L)):\ j=0,\ldots,L-1\bigr\},\\
 W_{L,d}&:=\bigl\{(\cos(2\pi j/L)\,v,\sin(2\pi j/L)):\
    v\in W_{L,d-1},\ j=0,\ldots,L-1\bigr\},\quad d\ge2.
 \end{aligned}
\]
Here $W_{L,d}$ is an indexed list of $L^d$ sampled vectors in $\mathbb R^{d+1}$; some vectors may represent the same projective ray. Accordingly, $N_W=L^d$ counts the sampled constraints with their original indexing, rather than the number of distinct rank-one projectors. This redundancy does not change the set of variables satisfying the sampled dual constraints. The worst-case projective overlap is
\[
 F(W):=\min_{\|x\|=1}\max_{w\in W}|\langle x,w\rangle|.
\]
The inductive estimate in Ref.~\cite[Sec.~5.4, Eq.~(107)]{HayashiOuyang2023} gives $F(W_{L,d})\ge\cos^d(\pi/L)$. Since normalized real vectors satisfy
\[
 \bigl\||x\rangle\langle x|-|w\rangle\langle w|\bigr\|_1
 =2\sqrt{1-|\langle x,w\rangle|^2},
\]
the covering radius in Eq.~\eqref{eq:projective-radius} obeys
\[
 \delta(W_{L,d})=2\sqrt{1-F(W_{L,d})^2}
 \le\bar\delta_{L,d}:=2\sqrt{1-\cos^{2d}(\pi/L)}
 =\frac{2\pi\sqrt d}{L}+O(L^{-3})
\]
for fixed $d$ as $L\to\infty$. This proves Eq.~\eqref{eq:reference-mesh-scaling} and supplies a computable upper bound on the covering radius; equality need not hold. All resolutions used in Sec.~\ref{sec:numerics} are even, and the numerical implementation uses $\bar\delta_{L,2}=2\sqrt{1-\cos^4(\pi/L)}$ rather than a separately optimized exact covering radius.

\subsection{Extension to arbitrary positive-definite weights}
\label{app:general-G}

Let $U_W^{(m)}[G]$ denote the optimal value of the finite-ray tight program with weight $G\succ0$; the main-text notation is $U_W^{(m)}=U_W^{(m)}[I_d]$. The identity-weight argument extends directly to fixed $G\succ0$. In affine coordinates the normalized scalar constraint becomes
\[
 q_{\lambda,u}^{G}(y)
 =y^TGy-b_{\lambda,u}^Ty-s_{\lambda,u},
\]
with minimizer $y_{\lambda,u}^{G}=\tfrac12G^{-1}b_{\lambda,u}$ and exact identity
\[
 q_{\lambda,u}^{G}(z)-q_{\lambda,u}^{G}(y_{\lambda,u}^{G})
 =(z-y_{\lambda,u}^{G})^TG(z-y_{\lambda,u}^{G}).
\]
Define $C_G(a):=\|G^{-1}\|_{\op}C_2^{\mathrm{HO}}(a)$.
The same block-norm estimate gives $\|y_{\lambda,u}^{G}\|\le C_G(a)$. The corresponding nonasymptotic extension is as follows.

\begin{proposition}[General positive-definite weight]
\label{prop:general-G}
Let $G\succ0$, and let $W$ be a mesh for which the finite-ray primal constraints admit a solution. Let $(a,\{\bar S_\lambda\})$ satisfy the $G$-weighted dual block positivity constraints at every sampled ray, as specified in Eq.~\eqref{eq:finite-ray-dual-rays}. Define $\varepsilon_W[G]:=U_W^{(m)}[G]-J_W^{(m)}(a,\bar S)\ge0$. If
\[
 \delta(W)<\frac{2}{\sqrt{1+C_G(a)^2}},
\]
then
\[
 0\le U_W^{(m)}[G]-C^{(m)}[G]
 \le
 \varepsilon_W[G]+\lambda_{\max}(G)\,
 \eta\!\left(C_G(a),\delta(W)\right)^2.
\]
Consequently,
\[
 J_W^{(m)}(a,\bar S)-\lambda_{\max}(G)\eta\!\left(C_G(a),\delta(W)\right)^2
 \le C^{(m)}[G]\le U_W^{(m)}[G].
\]
For fixed $m$ and fixed $G\succ0$, the difference between the upper and lower bounds is $O(\delta(W)^2)$ whenever the sampled dual constraints hold, $\varepsilon_W[G]=O(\delta(W)^2)$, and $C_G(a)$ remains bounded.
\end{proposition}

\begin{proof}
Set
\[
 t_G:=\lambda_{\max}(G)\,
 \eta\!\left(C_G(a),\delta(W)\right)^2.
\]
The radius estimate places every $y_{\lambda,u}^{G}$ in the Euclidean ball of radius $C_G(a)$. For any sector $\lambda$ and unit $u\in\cU_\lambda$, Lemma~\ref{lem:projective-affine} gives a sampled affine point $z$ with
\[
 \|z-y_{\lambda,u}^{G}\|\le
 \eta\!\left(C_G(a),\delta(W)\right).
\]
Sampled dual feasibility gives $q_{\lambda,u}^{G}(z)\ge0$. By the exact quadratic identity and the operator-norm bound for $G$,
\begin{align}
 q_{\lambda,u}^{G}(y_{\lambda,u}^{G})
 &=q_{\lambda,u}^{G}(z)
  -(z-y_{\lambda,u}^{G})^TG(z-y_{\lambda,u}^{G})\notag\\
 &\ge-\lambda_{\max}(G)\|z-y_{\lambda,u}^{G}\|^2
 \ge-t_G.
\label{eq:G-repair-minimum}
\end{align}
The point $y_{\lambda,u}^{G}$ minimizes $q_{\lambda,u}^{G}$, hence $q_{\lambda,u}^{G}(y)\ge-t_G$ for every $y\in\mathbb R^d$ and every unit $u$. Equivalently, the normalized matrix-valued affine constraint is bounded below by $-t_G I$ in every sector. Define the repaired completeness multipliers by
\begin{equation}
 \bar S_\lambda':=\bar S_\lambda-t_G\bar R_\lambda
 \qquad\text{for every }\lambda.
\label{eq:G-repaired-multiplier}
\end{equation}
In affine coordinates, replacing $\bar S_\lambda$ by $\bar S_\lambda'$ adds $t_G I$ to the normalized dual matrix. Thus the repaired dual constraints hold for all rays with $w_0\neq0$. If $w_0=0$, the dual matrix instead reduces to $(w_{1:d}^TGw_{1:d})\bar R_\lambda\succeq0$, which is unaffected by the choice of $\bar S_\lambda'$. Therefore $(a,\{\bar S_\lambda'\})$ is feasible for the full continuous dual.

Writing $p_\lambda=\Tr\bar R_\lambda$ and using $\sum_\lambda p_\lambda=1$, the repaired objective satisfies
\begin{align}
 \Tr a+\sum_\lambda\Tr\bar S_\lambda'
 &=J_W^{(m)}(a,\bar S)-t_G\sum_\lambda p_\lambda\notag\\
 &=J_W^{(m)}(a,\bar S)-t_G.
\label{eq:G-repair-objective}
\end{align}
Weak duality for the continuous problem and finite-ray primal feasibility now yield
\[
 U_W^{(m)}[G]-\varepsilon_W[G]-t_G
 \le C^{(m)}[G]\le U_W^{(m)}[G].
\]
This proves the two-sided bound and its stated error estimate. For fixed $G$ and bounded $C_G(a)$, Eq.~\eqref{eq:eta} gives $t_G=O(\delta(W)^2)$ as $\delta(W)\to0$, proving the final assertion.
\end{proof}

At the level of the weighted loss, one may also introduce the linear reparametrization $\widetilde\theta=G^{1/2}\theta$ to obtain identity weight. Proposition~\ref{prop:general-G} is stated directly in the original coordinates because it keeps the chosen projective mesh fixed and makes the resulting metric factors explicit.

\section{Feasibility and dual boundedness on sufficiently fine meshes}
\label{app:uniform-mesh}

This section proves the two properties used in Corollary~\ref{cor:second-order}. The first lemma constructs a finite-ray feasible point on every sufficiently fine mesh, with a positive lower bound on its selected measurement effects that is independent of the mesh. The second lemma uses that lower bound to control each sampled dual residual; differences between residuals at paired mesh directions then bound the dual coefficients. These are the two ingredients needed to make the quadratic repair estimate uniform over near-optimal sampled dual solutions.

Throughout, $m,d$ and the full-rank statistical model are fixed, $G=I_d$, and all sums over $\lambda$ run over the finite set $\Ymn$. Recall that
\[
 \sum_\lambda\Tr\bar R_\lambda=1,
 \qquad \sum_\lambda\Tr\bar D_{j,\lambda}=0.
\]
The second equality is a trace identity summed over sectors; it does not require $\Tr\bar D_{j,\lambda}=0$ for each sector.

\begin{lemma}[Feasibility on every sufficiently fine projective mesh]
\label{lem:fine-mesh-feasible}
There are constants $\delta_0>0$ and $c>0$ such that every finite indexed mesh $W$ with $\delta(W)<\delta_0$ admits a feasible solution of Eqs.~\eqref{eq:finite-completeness}--\eqref{eq:finite-LU} supported on at most $2d$ distinct projective directions. The corresponding nonzero normalized effects $E_{i,\sigma,\lambda}:=|w_{i,\sigma,0}|^2X_{i,\sigma,\lambda}$ satisfy $E_{i,\sigma,\lambda}\succeq cI_{\cU_\lambda}$, with $c$ independent of $W$.
\end{lemma}
\begin{proof}
We first construct sector operators dual to the tangent directions; these operators will be used to enforce the derivative constraints of the finite-ray primal. Define the real matrix
\[
 H_{jk}:=\sum_\lambda\Tr(\bar D_{j,\lambda}\bar D_{k,\lambda}).
\]
The exact Schur--Weyl block representation and Lemma~\ref{app:tangent-independence} imply linear independence of the families $(\bar D_{j,\lambda})_\lambda$. For real $v$, the quantity $v^THv$ is the sum of the traces of the squares of the Hermitian operators $\sum_jv_j\bar D_{j,\lambda}$. It vanishes only when $v=0$. Thus $H$ is positive definite. Define
\begin{equation}
 T_{j,\lambda}:=\sum_{k=1}^d(H^{-1})_{jk}\bar D_{k,\lambda},
 \qquad
 \sum_\lambda\Tr(\bar D_{k,\lambda}T_{j,\lambda})=\delta_{kj}.
\label{eq:mesh-gram-dual}
\end{equation}
We next select pairs of sampled directions near the $2d$ affine target points $\pm te_i$. The separation of each pair will form an invertible matrix, while the size of its inverse will keep the constructed effects positive. Let $T_*:=\max_\lambda(\sum_j\|T_{j,\lambda}\|^2)^{1/2}$, choose $t>2dT_*$, and then choose $0<e<t/(2\sqrt d)$. By Lemma~\ref{lem:projective-affine}, a sufficiently small projective covering radius supplies, for each $i$, sampled rays with nonzero zeroth components and affine coordinates $z_{i,+},z_{i,-}$ obeying
\[
 \|z_{i,+}-te_i\|\le e,
 \qquad \|z_{i,-}+te_i\|\le e.
\]
Indeed, the lemma's error $\eta(t,\delta)$ tends to zero as $\delta\to0$ and its strict domain condition eventually holds. The $2d$ target directions are separated by more than $2e$, so the sampled projective directions can be chosen distinct.

Define the real $d\times d$ matrices
\begin{equation}
 B_{ik}:=z_{i,+,k}-z_{i,-,k},
 \qquad M:=(B^T)^{-1}.
\label{eq:mesh-BM}
\end{equation}
Each row of $B-2tI_d$ has Euclidean norm at most $2e$. Hence
\[
 \|B-2tI_d\|\le2\sqrt d\,e<t,
 \qquad \|M\|\le b:=\frac1{2t-2\sqrt d\,e}<\frac1t.
\]
In particular $B$ is invertible and $B^TM=I_d$.

We now use the inverse matrix $M$ and the tangent-dual operators $T_{j,\lambda}$ to define the selected measurement effects. For every $i$, sign $\sigma\in\{+1,-1\}$, and sector $\lambda$, set
\begin{equation}
 E_{i,\sigma,\lambda}:=\frac1{2d}I_{\cU_\lambda}
 +\sigma\sum_{j=1}^dM_{ij}T_{j,\lambda}.
\label{eq:mesh-constructed-effects}
\end{equation}
The perturbation in this expression has norm at most $\|M\|T_*\le bT_*$. Thus
\[
 E_{i,\sigma,\lambda}\succeq cI_{\cU_\lambda},
 \qquad c:=\frac1{2d}-bT_*>0.
\]
The positivity bound has the same constant $c$ for every sufficiently fine mesh. Summing the positive and negative effects also gives the completeness constraint,
\[
 \sum_{i,\sigma}E_{i,\sigma,\lambda}=I_{\cU_\lambda}.
\]
It remains to verify the derivative constraints. The global trace identity and Eq.~\eqref{eq:mesh-gram-dual} yield
\[
 \sum_\lambda\Tr(\bar D_{j,\lambda}E_{i,\sigma,\lambda})
 =\sigma M_{ij}.
\]
Consequently, with $B^TM=I_d$,
\begin{align*}
 &\sum_{\lambda,i,\sigma}z_{i,\sigma,k}
 \Tr(\bar D_{j,\lambda}E_{i,\sigma,\lambda})\\
 &\quad=\sum_i(z_{i,+,k}-z_{i,-,k})M_{ij}
 =(B^TM)_{kj}=\delta_{kj}.
\end{align*}
Finally, convert these measurement effects into primal variables by putting $X_{i,\sigma,\lambda}:=E_{i,\sigma,\lambda}/|w_{i,\sigma,0}|^2$ at the selected rays and assigning zero operators to all remaining indexed rays. The completeness and derivative identities just proved coincide with Eqs.~\eqref{eq:finite-completeness}--\eqref{eq:finite-LU}, since $\langle w_{i,\sigma}|Q_k|w_{i,\sigma}\rangle=|w_{i,\sigma,0}|^2z_{i,\sigma,k}$. Thus the construction satisfies both primal constraints and the mesh-independent positivity bound.
\end{proof}

\begin{lemma}[Uniform bounds for near-optimal sampled dual solutions]
\label{lem:uniform-dual}
After decreasing $\delta_0$ if necessary, there exists a constant $A_0<\infty$ such that every mesh $W$ with $\delta(W)<\delta_0$ and every sampled dual-feasible pair $(a,\bar S)$ satisfying
\[
 0\le U_W^{(m)}-J_W^{(m)}(a,\bar S)\le1
\]
obey $\|a\|_{\mathrm F}\le A_0$. Consequently $C_2^{\mathrm{HO}}(a)$ and the norms of all $\bar S_\lambda$ are uniformly bounded over these meshes and pairs.
\end{lemma}
\begin{proof}
We first bound the objective of the feasible construction from the preceding lemma. This provides a mesh-independent upper bound on the primal--dual gap for any sampled dual solution satisfying the hypothesis. Use the $2d$ effects constructed in Lemma~\ref{lem:fine-mesh-feasible}. All the constants $t,e,b,c$ are independent of the mesh. Put $r:=t+e$; the selected points satisfy $\|z_{i,\sigma}\|\le r$. The primal objective $P_W(X)$ of this feasible construction obeys
\begin{equation}
 0\le P_W(X)
 =\sum_{\lambda,i,\sigma}\|z_{i,\sigma}\|^2
 \Tr(\bar R_\lambda E_{i,\sigma,\lambda})\le r^2=:K,
\label{eq:mesh-primal-bound}
\end{equation}
where the upper bound follows from $\sum_{i,\sigma}E_{i,\sigma,\lambda}=I$ and $\sum_\lambda\Tr\bar R_\lambda=1$.

For $z\in\mathbb R^d$, define the Hermitian matrices
\begin{equation}
 A_{k,\lambda}:=\sum_j a_{kj}\bar D_{j,\lambda},
 \qquad F_\lambda(z):=\|z\|^2\bar R_\lambda
 -\sum_kz_kA_{k,\lambda}-\bar S_\lambda.
\label{eq:mesh-dual-residual}
\end{equation}
Sampled dual feasibility gives $F_\lambda(z_{i,\sigma})\succeq0$. Substitution of the primal normalization and derivative constraints gives the exact identity
\begin{equation}
 P_W(X)-J_W^{(m)}(a,\bar S)
 =\sum_{\lambda,i,\sigma}\Tr\bigl(F_\lambda(z_{i,\sigma})E_{i,\sigma,\lambda}\bigr).
\label{eq:mesh-primal-dual-identity}
\end{equation}
Next we bound each residual separately. Since $U_W^{(m)}\le P_W(X)$ and the assumed gap is at most one, the nonnegative sum in Eq.~\eqref{eq:mesh-primal-dual-identity} is at most $K+1$. Every summand is nonnegative; moreover, $E_{i,\sigma,\lambda}\succeq cI$ makes it at least $c\Tr F_\lambda(z_{i,\sigma})$. Positivity therefore gives
\begin{equation}
 \|F_\lambda(z_{i,\sigma})\|
 \le\Tr F_\lambda(z_{i,\sigma})
 \le H_0:=(K+1)/c.
\label{eq:mesh-residual-bound}
\end{equation}

We now use pairs of residuals to bound the dual coefficient operators $A_{k,\lambda}$. Subtracting Eq.~\eqref{eq:mesh-dual-residual} at the two selected points for each $i$ eliminates $\bar S_\lambda$ and yields
\begin{equation}
 \sum_kB_{ik}A_{k,\lambda}
 =\bigl(\|z_{i,+}\|^2-\|z_{i,-}\|^2\bigr)\bar R_\lambda
 -F_\lambda(z_{i,+})+F_\lambda(z_{i,-}).
\label{eq:mesh-difference}
\end{equation}
Each right-hand side has norm at most $r^2R_*+2H_0$, with $R_*:=\max_\lambda\|\bar R_\lambda\|$. Inverting $B$ and using $\|B^{-1}\|\le b$ gives
\[
 \|A_{k,\lambda}\|\le A_*:=\sqrt d\,b(r^2R_*+2H_0).
\]
The tangent-dual operators in Eq.~\eqref{eq:mesh-gram-dual} then recover the scalar entries of $a$ from the bounded $A_{k,\lambda}$:
\[
 a_{kj}=\sum_\lambda\Tr(A_{k,\lambda}T_{j,\lambda}).
\]
With $T_1:=\max_j\sum_\lambda\|T_{j,\lambda}\|_1$, the trace inequality gives $|a_{kj}|\le A_*T_1$ and hence $\|a\|_{\mathrm F}\le dA_*T_1=:A_0$, independently of $W$.

Finally, with $a$ bounded independently of the mesh, we bound the remaining dual variables. Using Eq.~\eqref{eq:mesh-dual-residual} at one selected point, the bounds $\|z\|\le r$, $\|A_{k,\lambda}\|\le A_*$, and Eq.~\eqref{eq:mesh-residual-bound} give
\[
 \|\bar S_\lambda\|\le r^2R_*+\sqrt d\,rA_*+H_0.
\]
The compact-search radius $C_2^{\mathrm{HO}}(a)$, defined in Eq.~\eqref{eq:C2-HO}, is a continuous function of $a$ for the fixed full-rank model. It is therefore bounded when $\|a\|_{\mathrm F}\le A_0$.
\end{proof}

\clearpage
\section{Certified two-copy finite-ray SDP examples}
\label{app:certified-finite-ray}

We report certified enclosures for the finite-ray SDP optimum in two
instances of the primal and sampled dual defined in
Eqs.~\eqref{eq:finite-ray-primal},
\eqref{eq:finite-completeness}, and \eqref{eq:finite-LU}.
Both instances use the two-parameter qubit model
of Eq.~\eqref{eq:numerical-qubit-model}, the weight matrix $G=I_2$,
copy number $m=2$, and the projective reference mesh $W_{48,2}$.
These enclosures bound the finite-ray SDP optimum $U_{W_{48,2}}^{(2)}$;
they do not, by themselves, certify a lower bound for the continuum tight
value $C^{(2)}[I_2]$.

The following outward-rounded endpoints are valid:
\begin{equation}
\begin{array}{c|cc}
 r & \text{Certified lower bound for }U_{W_{48,2}}^{(2)}
   & \text{Certified upper bound for }U_{W_{48,2}}^{(2)} \\
 \hline
 0   & 1.501442773224 & 1.501443170418 \\
 0.5 & 1.258557482498 & 1.258557983087
\end{array}
\label{eq:certified-two-copy-finite-ray}
\end{equation}
The interval widths are less than $4.0\times10^{-7}$ and
$5.1\times10^{-7}$, respectively.
The endpoints are certified by outward interval evaluation and are not
merely the stopping tolerances of a numerical SDP solver.

The certificate for $r=0.5$ uses the triplet and singlet blocks
\begin{align}
 \bar R_t&=\operatorname{diag}(9,3,1)/16,
 & \bar R_s&=3/16,\\
 \bar D_{1,t}&=\operatorname{diag}(3,-1,-1)/4,
 & \bar D_{1,s}&=-1/4,\\
 \bar D_{2,t}&=\frac{1}{4\sqrt2}
 \begin{pmatrix}0&3&0\\3&0&1\\0&1&0\end{pmatrix},
 & \bar D_{2,s}&=0.
\end{align}
For the sampled dual, rationalized multiplier matrices are adjusted
to ensure strict feasibility.
At all $48^2=2304$ indexed mesh directions, the singlet constraint and the
three leading principal minors of the real symmetric triplet constraint are
positive. The trigonometric mesh entries are enclosed by rational intervals
obtained from a rational enclosure of $\pi$ and Taylor remainder bounds.
Consequently, the reported dual objective is a rigorous lower bound on
$U_{W_{48,2}}^{(2)}$.

The primal construction starts with ten rank-one rational triplet effects
and one nonnegative rational singlet effect, each assigned to a grid direction.  To enforce the linear
constraints exactly, the construction uses four additional triplet grid rays with
affine coordinates $(\pm t,0)$ and $(0,\pm t)$,
where $t=2+\sqrt3$, and a singlet ray at the origin.
Let
\begin{equation}
 T_1=\frac{16}{11}\bar D_{1,t},\qquad
 T_2=\frac85\bar D_{2,t}.
\end{equation}
These operators satisfy
$\operatorname{Tr}(\bar D_{j,t}T_k)=\delta_{jk}$.
Scale the original effects by $1-\tau$ with $\tau=10^{-8}$, and add
four uncorrected triplet anchor effects at the indicated rays, given by
\begin{equation}
 A_{i,\sigma}:=\frac14 I_t+\frac{\sigma}{2t}T_i,
 \qquad E^{(0)}_{i,\sigma,t}=\tau A_{i,\sigma}.
\label{eq:cert-scaled-anchors}
\end{equation}
Let $E^{(0)}_{s,\lambda}$ denote all the uncorrected effects after scaling
the original effects by $1-\tau$ and adding the anchors, including
the singlet anchor. With $z_s$ the affine coordinate of grid ray $s$,
define the triplet completeness residual and the derivative residuals by
\begin{align}
 \Delta C_t&:=\sum_s E^{(0)}_{s,t}-I_t,\label{eq:cert-triplet-residual}\\
 \Delta R_{ij}&:=\sum_{\lambda\in\{t,s\}}\sum_s
 z_{s,i}\Tr(\bar D_{j,\lambda}E^{(0)}_{s,\lambda})
 -\delta_{ij}.\label{eq:cert-derivative-residual}
\end{align}
Thus $\Delta R$ includes contributions from both the triplet and
singlet sectors. The triplet anchor at $\sigma t e_i$ is adjusted by
\begin{equation}
 H_{i,\sigma}=-\frac14\Delta C_t
 -\frac{\sigma}{2t}\sum_{j=1}^2\Delta R_{ij}T_j,
 \qquad \sigma\in\{+1,-1\}.
\label{eq:certified-triplet-anchor-correction}
\end{equation}
The correction obeys the two identities
\begin{align}
 \sum_{i=1}^2\sum_{\sigma=\pm1}H_{i,\sigma}
 &=-\Delta C_t,\label{eq:cert-anchor-completeness}\\
 \sum_{\sigma=\pm1}(\sigma t)\Tr(\bar D_{j,t}H_{k,\sigma})
 &=-\Delta R_{kj}.\label{eq:cert-anchor-derivative}
\end{align}
The first identity follows by summing over $\sigma$, and the second
uses $\Tr(\bar D_{j,t}T_k)=\delta_{jk}$.
The singlet normalization is corrected separately at the zero-cost
ray with affine coordinate zero; this does not change the derivative
residual. Thus completeness and the derivative constraints for local unbiasedness
hold exactly after correction.

For $r=0.5$, rational interval calculations bound each entry of the residuals
by less than $10^{-11}$. Since $\|T_i\|\le12/11$,
the unscaled anchor matrices defined in Eq.~\eqref{eq:cert-scaled-anchors}
satisfy $\lambda_{\min}(A_{i,\sigma})\ge1/4-(12/11)/(2t)>0$.
Consequently, the actual uncorrected anchor effects obey
\begin{equation}
 \lambda_{\min}(E^{(0)}_{i,\sigma,t})\ge
 \tau\left(\frac14-\frac{12/11}{2t}\right).
\label{eq:cert-scaled-anchor-margin}
\end{equation}
The norm of each correction $H_{i,\sigma}$ is strictly smaller than this
margin. The corrected singlet anchor is also positive. Thus every effect in the construction is positive semidefinite, and the
corrected primal objective gives a rigorous upper bound.

For $r=0$, the sampled-dual semidefinite inequalities at all 2304 indexed
directions, positive semidefiniteness of the primal effects, exact satisfaction
of completeness and derivative constraints after anchor correction, and
outward-rounded objective bounds are likewise established. The singlet block
does not contribute to this instance. Its four triplet anchors satisfy the
same correction identities, with the model-specific $T_i$ and a rational
interval positivity margin.

In both instances, numerical candidates are used to construct the certificates,
while feasibility and objective inequalities are established by rational
interval arithmetic rather than inferred from numerical solver tolerances.
These certifications apply only to the two stated finite-ray instances and do
not provide rigorous enclosures for the other copy numbers and parameter
choices reported in this work.

\section*{Acknowledgments}
This work was supported by the Guangdong Provincial Quantum Science Strategic Initiative (Grant No. GDZX2505003).

Generative AI tools were used to assist with drafting and revising parts of the manuscript, improving the clarity of the exposition, and generating code for figures designed by the authors. They also assisted with selected mathematical calculations and checks. The research questions, mathematical framework, main results, and proof strategy were developed by the authors. All AI-assisted text, figure code, and calculations were reviewed and revised where necessary by the authors, who take full responsibility for the contents of this manuscript.

\begingroup
\setlength{\parskip}{0.10em}

\endgroup

\end{document}